\documentclass[sigconf, nonacm]{acmart}

\usepackage[ruled,noend]{algorithm2e}
\usepackage{graphicx}
\usepackage{makecell}
\usepackage{subfigure}
\usepackage{tabularx} 
\usepackage{amsthm}
\usepackage{multirow}
\usepackage{array}
\usepackage{makecell}
\newtheorem{definition}{Definition}
\newtheorem{theorem}{Theorem}
\newtheorem{corollary}{Corollary}

\newcommand\vldbdoi{XX.XX/XXX.XX}

\newcommand\vldbvolume{20}
\newcommand\vldbissue{1}
\newcommand\vldbyear{2026}
\newcommand\vldbauthors{\authors}
\newcommand\vldbtitle{\shorttitle} 
\newcommand\vldbavailabilityurl{URL_TO_YOUR_ARTIFACTS}
\newcommand\vldbpagestyle{plain} 

\begin{document}
\title{A New Lower Bounding Paradigm and Tighter Lower Bounds for Elastic Similarity Measures}

\author{Zemin Chao}
\affiliation{%
  \institution{Harbin Institute of Technology}
  \city{Harbin}
  \state{Heilongjiang}
  \country{China}
}
\email{chaozm@hit.edu.cn}

\author{Boyu Xiao}
\affiliation{%
  \institution{Harbin Institute of Technology}
  \city{Harbin}
  \state{Heilongjiang}
  \country{China}
 }
\email{2022110524@stu.hit.edu.cn}

\author{Zitong Li}
\affiliation{%
  \institution{Harbin Institute of Technology}
  \city{Harbin}
  \state{Heilongjiang}
  \country{China}
 }
\email{25S103343@stu.hit.edu.cn}

\author{Zhixin Qi}
\affiliation{%
  \institution{Harbin Institute of Technology}
  \city{Harbin}
  \state{Heilongjiang}
  \country{China}
 }
\email{qizhx@hit.edu.cn}

\author{Xianglong Liu}
\affiliation{%
  \institution{Harbin Institute of Technology}
  \city{Harbin}
  \state{Heilongjiang}
  \country{China}
 }
\email{23S003029@stu.hit.edu.cn}

\author{Hongzhi Wang}
\authornote{Corresponding author.}
\affiliation{%
  \institution{Harbin Institute of Technology}
  \city{Harbin}
  \state{Heilongjiang}
  \country{China}
 }
\email{wangzh@hit.edu.cn}

\begin{abstract}
Elastic similarity measures are fundamental to time series similarity search because of their ability to handle temporal misalignments. 
These measures are inherently computationally expensive, therefore necessitating the use of lower bounds to prune unnecessary comparisons. 
This paper proposes a new \emph{Bipartite Graph Edge-Cover Paradigm} for deriving lower bounds, which applies to a broad class of elastic similarity measures.  
This paradigm formulates lower bounding as a vertex-weighting problem on a weighted bipartite graph induced from the input time series.
Under this paradigm, most of the existing lower bounds of elastic similarity measures can be viewed as simple instantiations.
We further propose \textit{BGLB}, an instantiation of the proposed paradigm that incorporates an additional augmentation term, yielding lower bounds that are provably tighter.
Theoretical analysis and extensive experiments on 128 real-world datasets demonstrate that \textit{BGLB} achieves the tightest known lower bounds for six elastic measures (ERP, MSM, TWED, LCSS, EDR, and SWALE). 
Moreover, \textit{BGLB} remains highly competitive for \textit{DTW} with a favorable trade-off between tightness and computational efficiency.
In nearest neighbor search, integrating \textit{BGLB} into filter pipelines consistently outperforms state-of-the-art methods, achieving speedups ranging from $24.6\%$ to $84.9\%$ across various elastic similarity measures.
Besides, \textit{BGLB} also delivers a significant acceleration in density-based clustering applications, validating the practical potential of \textit{BGLB} in time series similarity search tasks based on elastic similarity measures.


\end{abstract}

\maketitle

\pagestyle{\vldbpagestyle}
\begingroup\small\noindent\raggedright\textbf{PVLDB Reference Format:}\\
\vldbauthors. \vldbtitle. PVLDB, \vldbvolume(\vldbissue):  \vldbyear.\\
\endgroup
\begingroup
\renewcommand\thefootnote{}\footnote{\noindent
This work is licensed under the Creative Commons BY-NC-ND 4.0 International License. Visit \url{https://creativecommons.org/licenses/by-nc-nd/4.0/} to view a copy of this license. For any use beyond those covered by this license, obtain permission by emailing \href{mailto:info@vldb.org}{info@vldb.org}. Copyright is held by the owner/author(s). Publication rights licensed to the VLDB Endowment. \\
\raggedright Proceedings of the VLDB Endowment, Vol. \vldbvolume, No. \vldbissue\ %
ISSN 2150-8097. \\
\href{https://doi.org/\vldbdoi}{doi:\vldbdoi} \\
}\addtocounter{footnote}{-1}\endgroup

\ifdefempty{\vldbavailabilityurl}{}{
\vspace{.3cm}
\begingroup\small\noindent\raggedright\textbf{PVLDB Artifact Availability:}\\
The source code, data, and/or other artifacts have been made available at \url{https://github.com/BoyuXiao/BGLB}.
\endgroup
}

\section{Introduction}

Time series similarity measures are fundamental building blocks in data analysis and data mining~\cite{10.14778/1454159.1454226}.
Their core objective is to quantify similarity between time series, enabling downstream workloads such as similarity search~\cite{10.1145/2020408.2020607,10.1145/253262.253332,10.1145/2000824.2000827},
classification~\cite{10.1007/s10618-016-0483-9},
clustering~\cite{10.14778/2735479.2735481,10.14778/3342263.3342648,10.1145/2949741.2949758},
motif discovery~\cite{10.14778/2735471.2735476,7837992},
anomaly detection~\cite{10.14778/3551793.3551830,breunig2000lof},
and broader data mining applications~\cite{5693959}.
Elastic similarity measures constitute one of the most important classes of similarity measures, as they effectively handle temporal distortions, noises and unequal time series lengths through non-linear point-wise mapping.~\cite{10.5555/3000850.3000887, 10.5555/1316689.1316758,10.1145/1066157.1066213,10.14778/1454159.1454226,10.1145/1247480.1247544,10.5555/108235.108244,10.1145/3318464.3389760,10.14778/3594512.3594530}. 

The general principle underlying elastic measures is to compute the similarity between two time series by identifying the optimal alignment between them, which is typically formulated as a minimal-cost sequence of alignment operations.
This process naturally involves expensive dynamic programming procedures, whose costs are proportional to the square of the time series length.
Moreover, theoretical outcomes established on the Strong Exponential Time Hypothesis (SETH) suggest that substantially faster algorithms to compute exact \textit{DTW} or \textit{LCSS} are unlikely to exist~\cite{7354388}.
Consequently, pairwise comparisons based on elastic similarity measures are computationally expensive for large-scale or long time series, motivating effective pruning techniques based on fast lower bounds.

Fortunately, in many time series data mining applications, exact similarity need only be evaluated for potentially similar candidate time series (pairs).
Lower bounds provide an efficient mechanism to filter out dissimilar candidates:
if a cheap-to-compute distance lower bound already exceeds a certain threshold, the candidate can be discarded without unnecessary evaluation of similarity measures.
This pruning strategy has been highly successful for \textit{DTW}, achieving speedups of several orders of magnitude compared to the naive approach in the past decades~\cite{10.1145/2339530.2339576}.

In summary, effective lower bounds are essential for the efficient computation of elastic similarity measures.  
However, existing lower bounds typically model the alignment between two time series \(X \in \mathbb{R}^n\) and \(Q \in \mathbb{R}^m\) as finding an optimal path in the \(m \times n\) cost matrix, and derive a lower bound by aggregating locally minimal admissible costs of the path along rows or columns.  
Although this paradigm has dominated research on elastic distance lower bounds over the past three decades, we reveal a fundamental limitation of it: row-wise (or column-wise) minimization cannot simultaneously ensure coverage of all columns (or rows). Consequently, the resulting bound ignores the cost associated with uncovered columns (or rows), leading to suboptimal lower bounds. See Section~\ref{sec:2-2} for a detailed discussion.

To address the aforementioned limitation, this paper proposes a fundamentally new paradigm for lower bounding elastic similarity measures.
Rather than using traditional dynamic programming cost matrix, we solve the problem with an induced bipartite graph, where vertices correspond to elements of two time series $X$ and $Q$, and the edge weights encode the costs of alignment operations (e.g. matching, insertion, and deletion) as defined by the underlying elastic measure.

We prove that any assignment of non-negative weights to the vertices of this graph that satisfies a set of simple, measure-specific constraints yields a valid lower bound on the true elastic distance.  
This formulation unifies lower bound design: different feasible vertex weightings give rise to different lower bounds, and notably, most existing methods can be viewed as special cases of this general paradigm.

Building upon this paradigm, we design \textit{BGLB} (Bipartite Graph-based Lower Bound),
which achieves the tightest known bounds for elastic similarity measures such as  \textit{LCSS}~\cite{994784},  \textit{EDR}~\cite{10.1145/1066157.1066213},  \textit{SWALE}~\cite{10.1145/1247480.1247544},  \textit{ERP}, \textit{MSM}~\cite{6189346} and \textit{TWED}.
For \textit{DTW}, while certain DTW-specific bounds can be slightly tighter, 	\textit{BGLB} remains among the strongest general-purpose bounds and provides a favorable tightness--efficiency trade-off. 

Experimental results demonstrate that \textit{BGLB} significantly accelerates elastic similarity search.
Integrating \textit{BGLB} into nearest neighbor search consistently yields speedups of $24.6\%$ to $84.9\%$ over the state of the art across a wide range of real-world datasets and elastic similarity measures.
Furthermore, \textit{BGLB} also demonstrates significant benefits in density-based clustering applications such as DBSCAN, underscoring its effectiveness in diverse tasks that rely on elastic similarity search.

The main contributions of this paper are summarized as follows:
\begin{enumerate}
\item 
This paper presents a new paradigm for computing the lower bounds of various elastic similarity measures.
This paradigm reformulates the lower bounding problem as a vertex-weighting task on a weighted bipartite graph induced by the input time series, where edge weights encode matching and deletion costs.
We prove that the lower bound of elastic similarity measures can be casted as a vertices assigning problem on the induced graph.
This graph-based paradigm subsumes most existing lower bounds as special cases.

\item
The proposed paradigm ensures that lower bound explicitly accounts for the coverage of every element in both time series, thereby avoiding the pitfall of traditional matrix-based paradigms, where certain rows or columns (which fundamentally correspond to time series elements) may be overlooked. 
Therefore, the new paradigm leads to more effective lower bounds for various elastic similarity measures.

\item 
Based on the above paradigm, we propose \textit{BGLB} (Bipartite Graph-based Lower Bound), a unified and tighter lower bound applicable across a broad class of elastic similarity measures.  
\textit{BGLB} can be viewed as an instantiation that enhances the traditional lower bound by incorporating an additional augmentation term.
Theoretically, \textit{BGLB} is provably tighter than the state-of-the-art general-purpose bound \textit{GLB} and achieves the tightest known lower bounds for ERP, MSM, TWED, LCSS, EDR, and SWALE, while remaining highly competitive for DTW with a favorable trade-off between tightness and computational efficiency.

\item 
Extensive experiments on 128 real-world UCR datasets demonstrate that integrating \textit{BGLB} into similarity search pipelines yields consistent and substantial performance gains. In 1-NN classification tasks, BGLB delivers speedups ranging from $24.6\%$ to $84.9\%$ over existing methods across seven elastic measures. Furthermore, in density-based clustering (e.g., exact DBSCAN), \textit{BGLB} accelerates end-to-end runtime by up to $29\%$ compared to \textit{GLB}, without compromising result accuracy, thereby validating its practical utility in real-world time series analytics.

\end{enumerate}

\section{Preliminaries and Related Work}
\subsection{Key Elastic Similarity Measures}
\label{2-1}

Recent studies show that no single elastic similarity measure consistently achieves the best performance across all time series datasets~\cite{10.1145/3318464.3389760}. 

For any two series $X \in \mathbb{R}^n$, $Q \in \mathbb{R}^m$, elastic measures compute a Raw Accumulated Cost matrix $RAC$ via dynamic programming, where diagonal/vertical/horizontal moves correspond to three cost functions, and the final distance is obtained by transforming the value obtained from the raw accumulated cost $RAC(n,m)$~\cite{10.14778/3594512.3594530}:
\begin{equation}
\label{eq:elastic-general}
RAC(i,j)=
\begin{cases}
\textit{initial\_distance}(x_1,q_1), & i=j=1\\
RAC(i-1,j)+\mathbf{D_V}(x_i,q_1), & i\neq 1,\, j=1\\
RAC(i,j-1)+\mathbf{D_H}(x_1,q_j), & i=1,\, j\neq 1\\
\min\!\left\{
\begin{array}{l}
RAC(i-1,j-1)+\mathbf{M}(x_i,q_j)\\
RAC(i-1,j)+\mathbf{D_V}(x_i,q_j)\\
RAC(i,j-1)+\mathbf{D_H}(x_i,q_j)
\end{array}\right., & i\neq 1,\, j\neq 1
\end{cases}
\end{equation}
where \(D_H\), \(D_V\), and \(M\) are the costs depending on the underlying elastic distance measures. 
The elastic distance is defined as $\mathrm{ELD}(X,Q)=\mathrm{trans}(RAC(n,m))$.
The specific cost for the elastic measures considered in this paper are summarized in Table \ref{tab:elastic-measures-block}.

\newlength{\wRow}  \setlength{\wRow}{10mm}      
\newlength{\wDTW}  \setlength{\wDTW}{12mm}
\newlength{\wLCSS} \setlength{\wLCSS}{19mm}
\newlength{\wEDR}  \setlength{\wEDR}{19mm}
\newlength{\wSWA}  \setlength{\wSWA}{19mm}
\newlength{\wERP}  \setlength{\wERP}{12mm}
\newlength{\wMSM}  \setlength{\wMSM}{42mm}
\newlength{\wTWED} \setlength{\wTWED}{32mm}

\begin{table*}[t]
\caption{Movement cost functions (match/insert/delete) and the final transformation for seven elastic measures.}
\label{tab:elastic-measures-block}
\centering
\fontsize{6}{8}\selectfont
\setlength{\tabcolsep}{2pt}
\renewcommand{\arraystretch}{1.2}

\begin{tabular}{|
  >{\centering\arraybackslash}p{\wRow} |
  >{\centering\arraybackslash}p{\wDTW} |
  >{\centering\arraybackslash}p{\wLCSS}|
  >{\centering\arraybackslash}p{\wEDR} |
  >{\centering\arraybackslash}p{\wSWA} |
  >{\centering\arraybackslash}p{\wERP} |
  >{\centering\arraybackslash}p{\wMSM} |
  >{\centering\arraybackslash}p{\wTWED}|
}

\hline
& \multicolumn{1}{c|}{\textbf{DTW}}
& \multicolumn{1}{c|}{\textbf{LCSS}}
& \multicolumn{1}{c|}{\textbf{EDR}}
& \multicolumn{1}{c|}{\textbf{SWALE}}
& \multicolumn{1}{c|}{\textbf{ERP}}
& \multicolumn{1}{c|}{\textbf{MSM}}
& \multicolumn{1}{c|}{\textbf{TWED}} \\
\hline

$\mathbf{M}(x_i,q_j)$
& $(x_i-q_j)^2$
& $\begin{cases}
1 & \text{if } |x_i-q_j|\le \varepsilon\\
0 & \text{otherwise}
\end{cases}$
& $\begin{cases}
0 & \text{if } |x_i-q_j|\le \varepsilon\\
1 & \text{otherwise}
\end{cases}$
& $\begin{cases}
r & \text{if } |x_i-q_j|\le \varepsilon\\
p & \text{otherwise}
\end{cases}$
& $(x_i-q_j)^2$
& $|x_i-q_j|$
& $\begin{array}{l}
|x_i-q_j|+|x_{i-1}-q_{j-1}|\\
 +\nu\big(|t^x_i-t^x_{i-1}|+|t^q_j-t^q_{j-1}|\big)
\end{array}$
\\ \hline

$\mathbf{D_V}(x_i,q_j)$
& $(x_i-q_j)^2$
& $0$
& $1$
& $p$
& $(x_i-g)^2$
& $\begin{cases}
c, \text{if } q_{j-1}\le q_j\le x_i \ \text{or}\ q_{j-1}\ge q_j\ge x_i\\
c+\min\{|q_j-q_{j-1}|,\ |q_j-x_i|\},   \text{otherwise}
\end{cases}$
& $|x_i-x_{i-1}|+\nu|t^x_i-t^x_{i-1}|+\lambda$
\\ \hline

$\mathbf{D_H}(x_i,q_j)$
& $(x_i-q_j)^2$
& $0$
& $1$
& $p$
& $(q_j-g)^2$
& $\begin{cases}
c, \text{if } x_{i-1}\le x_i\le q_j \ \text{or}\ x_{i-1}\ge x_i\ge q_j\\
c+\min\{|x_i-x_{i-1}|,\ |x_i-q_j|\},  \text{otherwise}
\end{cases}$
& $|q_j-q_{j-1}|+\nu|t^q_j-t^q_{j-1}|+\lambda$
\\ \hline

$\mathbf{trans()}$
& $\sqrt{RAC(n,m)}$
& $1-\dfrac{RAC(n,m)}{\min(n,m)}$
& $RAC(n,m)$
& $RAC(n,m)$
& $RAC(n,m)$
& $RAC(n,m)$
& $RAC(n,m)$
\\ \hline

\end{tabular}
\end{table*}

In fact, the proposed approach is applicable to any elastic similarity measure that conforms to Equation~\eqref{eq:elastic-general}.  
However, due to space constraints, we restrict our detailed discussion to the most widely adopted elastic measures.  
Specifically, we focus on the following seven representative elastic similarity measures:

\noindent \textbf{Dynamic Time Warping (DTW)}~\cite{Sakoe1971ADP,1163055} is one of the most widely used similarity measures for time series, with applications spanning domains such as speech processing, biomedical signal analysis, and sensor monitoring.

\noindent \textbf{Longest Common Subsequence (LCSS)}~\cite{994784} was first developed for pattern matching over text strings and later adapted to quantify similarity between time series. 
LCSS increases the similarity score by $1$ when two elements are deemed to match, and by $0$ otherwise. 
The LCSS distance is obtained by normalizing the LCSS similarity by the length of the shorter time series.

\noindent \textbf{Edit Distance on Real Sequences (EDR)}~\cite{10.1145/1066157.1066213} is an adaptation of string edit distance to real-valued sequences. 
EDR computes the distance by assigning penalties to mismatches, rather than awarding a score for matches.

\noindent \textbf{Sequence Weighted Alignment (SWALE)}~\cite{10.1145/1247480.1247544} generalizes EDR by introducing two parameters: a match reward $r$ and a mismatch penalty $p$, replacing the fixed match/mismatch scores used in EDR.

\noindent \textbf{Edit Distance with Real Penalty (ERP)}~\cite{10.5555/1316689.1316758} is closely related to DTW in that it uses the squared difference between two aligned elements, $(x_i-q_j)^2$, as the diagonal (match) cost. Unlike DTW, ERP introduces a gap/reference value $g$ to define the costs of horizontal and vertical moves.

\noindent \textbf{Move-Split-Merge (MSM)}~\cite{6189346} integrates desirable properties from multiple elastic measures and is a translation-invariant metric. MSM assigns a constant cost $c$ to the \textit{split} operation (replicating the previous element) and the \textit{merge} operation (combining two identical elements into one), which correspond to vertical and horizontal moves in the distance matrix. The \textit{move} operation corresponds to a diagonal move, with cost $|x_i-q_j|$.

\noindent \textbf{Time Warp Edit Distance (TWED)}~\cite{4479483} penalizes differences in timestamps in addition to differences in numeric values. TWED uses a parameter $v$ to penalize timestamp discrepancies in all three move types, and additionally applies a stiffness parameter $\lambda$ in horizontal and vertical moves to control the amount of warping.

\subsection{Summary of Existing Lower Bound Methodologies}
\label{sec:2-2}

There exists a substantial number of lower bounds methods for different elastic similarity measures~\cite{10.5555/2993953.2994027,KEOGH2002406,914875,doi:10.1137/1.9781611975321.27,LEMIRE20092169,WEBB2021107895,10.14778/3594512.3594530,10.5555/1316689.1316758,Tan2020FastEE,10.1145/956750.956777,Vlachos2006Indexing}.  
Due to space limitations, we do not provide detailed descriptions of these lower bounds. 
Nevertheless, a common characteristic unifies most of these existing approaches: they rely on a strip-based strategy to derive lower bounds for elastic similarity measures.

We use \( \text{LB}_{\text{Keogh}} \) as a representative example to illustrate the conventional strip-based paradigm. In this approach, an \( n \times m \) cost matrix is constructed, as shown in Figure~\ref{fig_old_lb}, where the entry at row \( i \) and column \( j \) represents the matching cost between \( q_i \) and \( x_j \). For DTW, this cost is typically \( (x_j - q_i)^2 \).

\begin{figure}[h]
  \centering
    \includegraphics[width=\linewidth]{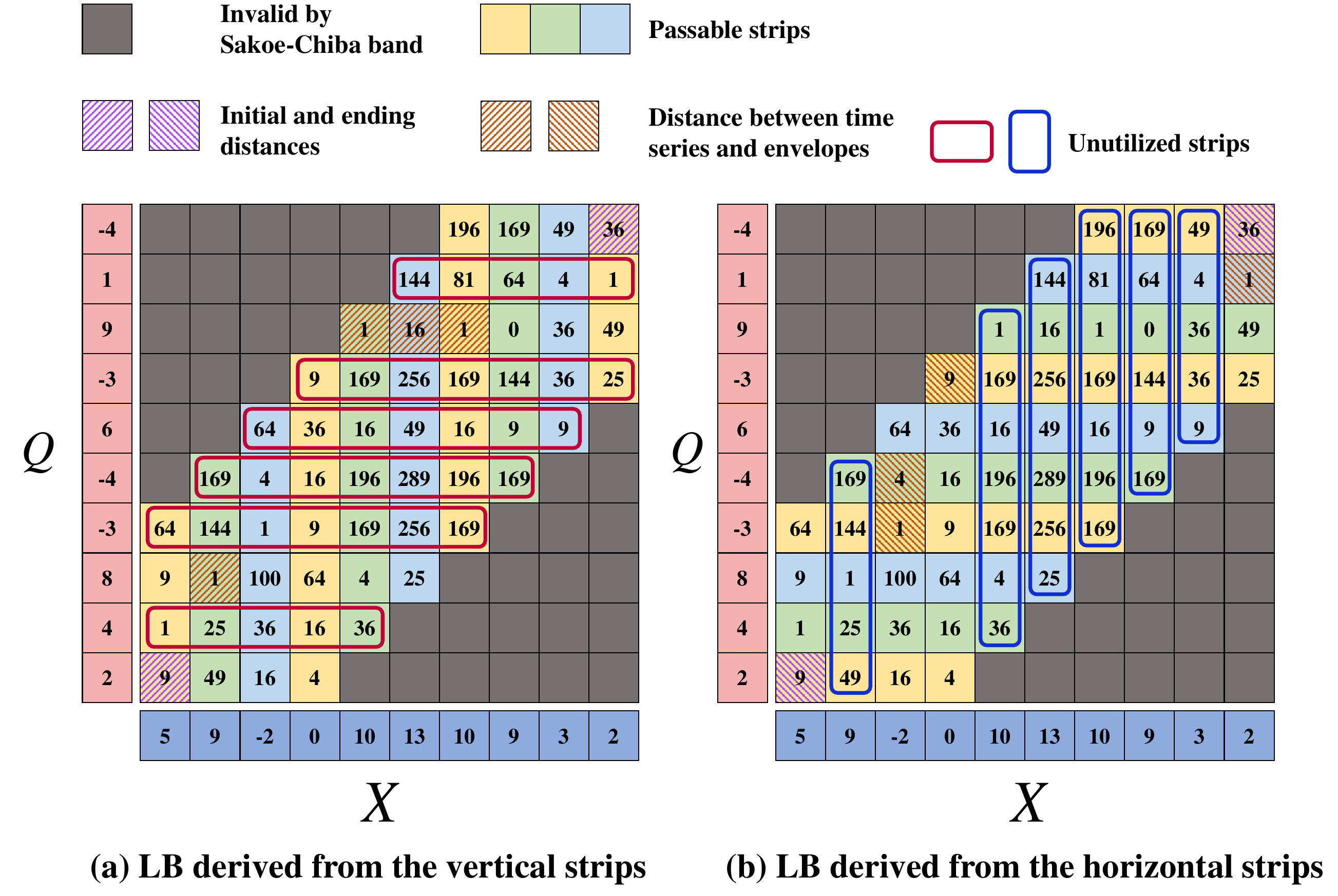}
  \caption{
Illustration of the strip-based methodology used by existing lower bounds.
Such methods typically take the minimum cost from each vertical (or horizontal) strip, potentially leaving the other dimension underutilized.
}
  \Description{Description of Existing Lower Bounds.}
  \label{fig_old_lb}
\end{figure}

The exact elastic distance corresponds to the minimum-cost path from the bottom-left to the top-right corner of the cost matrix.  
To obtain a valid lower bound i.e., an underestimate that never exceeds the true distance, existing methods partition the matrix into non-overlapping vertical strips (Figure~\ref{fig_old_lb}(a)) or horizontal strips (Figure~\ref{fig_old_lb}(b)). 

Since traversing each strip incurs a non-negative cost that depends on where the path crosses it, one can compute the minimum possible cost to cross each individual strip. The total cost of any valid path must therefore be at least the sum of these per-strip minima. Consequently, the strip-based paradigm derives lower bounds for elastic distance measures by summing the minimal admissible cost within each strip.

This paradigm underlies most existing lower bounds of elastic similarity measures, including \( \text{LB}_{\text{Keogh}} \)  and \textit{GLB}~\cite{10.14778/3594512.3594530}.
However, this paradigm suffers from a fundamental limitation: 
it computes only local minima along vertical or horizontal strips, potentially overlooking coverage requirements in the orthogonal direction.
As shown in Figure~\ref{fig_old_lb}(a), while the cost of passing each vertical strips is fully accounted for, many horizontal strips (highlighted by red rounded rectangles) remain uncovered and contribute zero to the bound—even though any valid warping path must incur a non-negative cost to pass through every horizontal strip. 
Similarly, horizontal strip partitioning may omit the cost of passing vertical strips (blue borders in Figure~\ref{fig_old_lb}(b)).
This is logically equivalent to the elements corresponding to the ignored horizontal or vertical strips being omitted and leads to suboptimal lower bounds.

The true warping path in elastic measures is required to be both continuous and monotonic, which requires that every row (horizontal strips) and every column (vertical strips) of the cost matrix be visited. The strip-based construction, however, does not enforce this global coverage constraint. This mismatch between the lower bound formulation and the structural properties of optimal alignment paths leads to suboptimal tightness and leaves significant room for improvement.

This occurs because multiple strips may select cells from the same row (or column), leaving other rows (or columns) unused. 
Specifically, for vertical strips, each element \( x_j \in X \) is associated with a passable region defined by the warping window around \( q_i \), and it may contribute to at most one cell in its corresponding column. The minimal cost in each vertical strip is identified as the lowest-cost cell within that strip; these cells are marked with brown diagonal stripes in Figure~\ref{fig_old_lb}. 
Clearly, neither horizontal nor vertical strips alone can resolve this issue.

To overcome the limitation above, this paper proposes a new paradigm to compute the lower bounds.
We reformulate the elastic measures lower bounding problem as an assignment problem on the induced bipartite graph: vertices on one side represent elements of time series \( X \), and those on the other side represent elements of time series \( Q \); edges between them are weighted by the cost of corresponding alignment operations. 
The problem then reduces to computing a minimum-weight edge cover on this bipartite graph, which inherently enforces full coverage of every element in both time series when estimating the lower bounds.

\begin{figure*}[t]
  \centering
  \includegraphics[width=\linewidth]{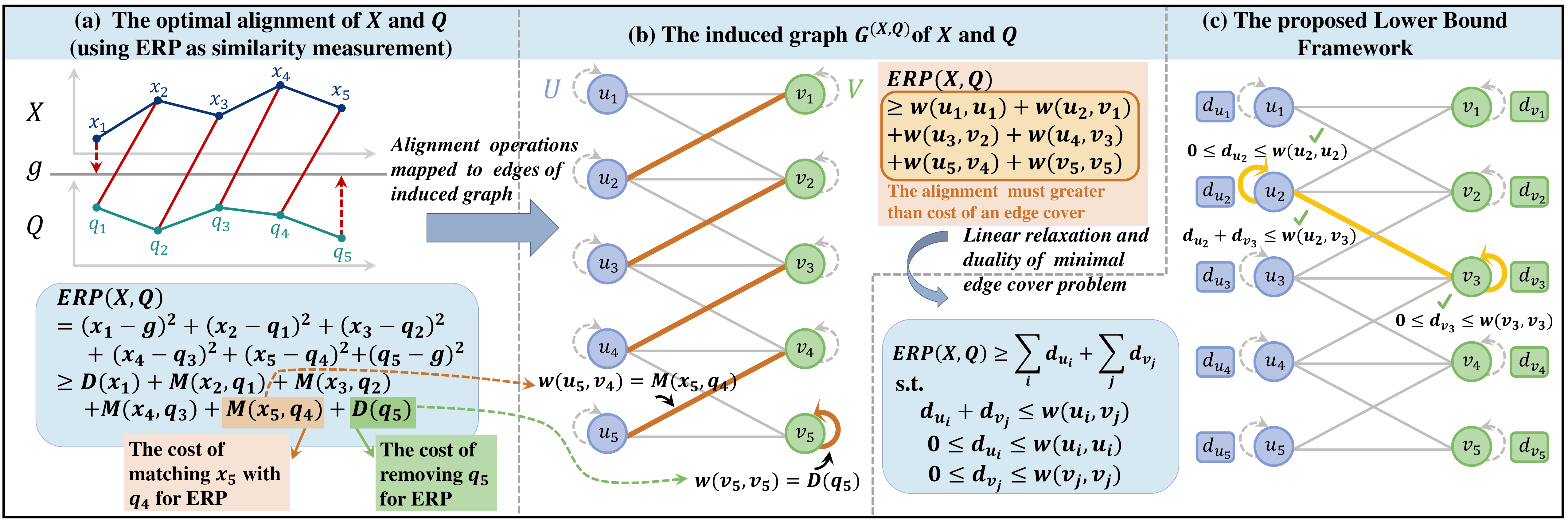}
  \caption{Overview of the Bipartite Graph based Framework, where the lower bound of elastic similarity measure is mapped to the assigning on the Induced graph $G^{(X,Q)}$.(a)The optimal alignment of $X$ and $Q$ under ERP. (b)The corresponding vertex cover on induced bipartite graph $G^{(X,Q)}$. (c) The valid assignment based on linear-dual and relation.}
  \label{fig:erp_overview}
\end{figure*}

\section{A Unified Lower Bound Framework Based on Induced Graphs}

This section introduces our \emph{Bipartite Graph Edge-Cover Paradigm} for deriving lower bounds of elastic similarity measures.
Figure~\ref{fig:erp_overview} illustrates the methodology of this proposed paradigm.
First, we simplify the general formulation of elastic distance measures by abstracting the alignment of two time series as a sequence of alignment operations consisting only of element-wise matching and (abstracted) element deletion (Section \ref{subsection_overview}).  
Building upon this unified view, Sections \ref{sec_defination_of_G} and \ref{section_mapping_lb_to_G} introduce the induced graph of the input time series $X$, $Q$ and prove that computing a lower bound on elastic similarity measures can be mapped to a minimum-weight edge cover problem on this graph.
Finally, by applying linear programming relaxation and duality, we establish a new computational paradigm for lower bounding elastic distances.
That is, for any elastic measure to satisfy the definition in Equation \eqref{eq:elastic-general}, the problem of computing its lower bound reduces to a vertex-weighting problem on the induced graph (Section \ref{section_releasing_the_problem}).
Within this paradigm, existing generalized bounds such as \textit{GLB} can be interpreted as particular instantiations, and the next section derives \textit{BGLB} as a tighter instantiation.

\subsection{The Observations on Elastic Similarity Measures}
\label{subsection_overview}

All the elastic similarity measures mentioned in Section~\ref{2-1} can be interpreted in terms of \emph{matching}, \emph{insertion} and \emph{deletion} operations between two time series.
Although insertion operations are formally defined in some measures, any insertion on one sequence can be replaced with a deletion on the other \footnote{Essentially, this is due to the symmetry inherent in elastic distance measures, which requires that the costs of inserting and deleting elements be symmetric.}.
That is, an insertion on $X$ can be viewed as deleting the corresponding element from $Q$.

Let $D(x_i)$ denote the least cost to "delete" $x_i$ from $X$, which is, 
\begin{equation}
 D(x_i)=\min\limits_{j\in \{i-\omega,\dots ,i+\omega\}} \{D_V(x_i,q_j), D_H(x_i,q_j) \}.   
\end{equation}

Then the following inequality holds for the raw accumulated cost $RAC$:


\begin{equation}
\label{eq:unified_two_ops}
RAC(i,j) \geq
\begin{cases}
\textit{initial\_distance}(x_1,q_1), & i=j=1\\
RAC(i-1,j)+\mathbf{D}(x_i), & i\neq 1,\, j=1\\
RAC(i,j-1)+\mathbf{D}(q_j), & i=1,\, j\neq 1\\
\min\!\left\{
\begin{array}{l}
RAC(i-1,j-1)+\mathbf{M}(x_i,q_j)\\
RAC(i-1,j)+\mathbf{D}(x_i)\\
RAC(i,j-1)+\mathbf{D}(q_j)
\end{array}\right., & i\neq 1,\, j\neq 1
\end{cases}
\end{equation}
where $M(x_i,q_j)$ denotes the cost to match $x_i$ and $q_j$.


For two time series $X$ and $Q$,
our observations on elastic similarity measures are as follows:

\noindent(1) Every element $x_i \in X$ must be either deleted (paying a cost $D(x_i)\geq 0$) or matched with some element $q_j \in Q$ (with a cost $M(x_i,q_j) \geq 0$).

\noindent(2) Equally, every element $q_j \in Q$ should be deleted (paying a cost $D(q_j)\geq 0$) or matched with some element $x_i \in X$ (with a cost $M(x_i,q_j) \geq 0$).

\noindent(3)There exists at least one sequence of match or delete operations that aligns \(X\) and \(Q\), and \(RAC(n,m)\) is the upper bound of the minimum total cost on all such operation sequences.

These observations lead directly to Corollary~\ref{corollary_obeservation}.

\begin{corollary}
\label{corollary_obeservation}
Given any time series $X$, $Q$ and elastic similarity measure, there exists a set of basic operations $\mathbb{S}$, such that

\begin{center}
$D(x_i) \in \mathbb{S}$ or $M(x_i,q_j) \in \mathbb{S}$ for every $x_i \in X$,
\end{center}
\begin{center}
$D(q_j) \in \mathbb{S}$ or $M(x_i,q_j) \in \mathbb{S}$ for every $q_j \in Q$,
\end{center}

and
\begin{equation}
RAC(n,m) \ge \sum_{s\in \mathbb{S}} Cost(s),
\end{equation}
where $Cost(s)$ is the cost of the corresponding basic operation.
\end{corollary}

Clearly, the costs of deletion and matching operations vary across different elastic measures; the specific definitions for each measure are provided in Table~\ref{table_DEL_and_matching_cost}.

\begin{table}[htbp]
\centering
\caption{Matching Cost and Generalized Deletion Cost for Different Elastic Similarity Measures}
\label{table_DEL_and_matching_cost}
\scalebox{0.88}{
\begin{tabular}{|c|c|c|}
\hline
\textbf{Method} & \textbf{Matching Cost } $M(x_i, q_j)$ & \textbf{Deletion Cost} $D(x_i)$ \\
\hline
\textit{MSM} & $|x_i - q_j|$ & $c$ \\
\hline
\textit{ERP} & $(x_i - q_j)^2$ & $(x_i - g)^2$ \\
\hline
\textit{TWED} & $\begin{array}{l}
|x_i-q_j|+|x_{i-1}-q_{j-1}|\\
 +\nu\big(|t^x_i-t^x_{i-1}|+|t^q_j-t^q_{j-1}|\big)
\end{array}$ & 
    $|x_i-x_{i-1}|+\nu|t^x_i-t^x_{i-1}|+\lambda$ \\
\hline
\textit{DTW} & $(x_i - q_j)^2$ & $ \min\limits_{j\in \{i-\omega,i+\omega\}}(x_i - q_j)^2$ \\ 
\hline
\textit{LCSS} & $\begin{cases}
        0 & \text{if } |x_i - q_j| \leq \epsilon \\
        1 & \text{otherwise}
    \end{cases}$ & 1 \\
\hline
\textit{EDR} & $\begin{cases}
        0 & \text{if } |x_i - q_j| \leq \epsilon \\
        1 & \text{otherwise}
    \end{cases}$ & 1 \\
\hline
\textit{SWALE} & $\begin{cases}
        r & \text{if } |x_i - q_j| \leq \epsilon \\
        p & \text{otherwise}
    \end{cases}$ & $p$ \\
\hline
\end{tabular}
}
\end{table}

\subsection{The Induced Graph of Elastic Measures}
\label{sec_defination_of_G}

\begin{table}[t]
\fontsize{5.5}{8}\selectfont 
\centering
\caption{Important Notations.}
\label{tab:notation}
\setlength{\tabcolsep}{4pt}
\renewcommand{\arraystretch}{1.15}
\begin{tabularx}{\linewidth}{|>{\centering\arraybackslash}p{0.25\linewidth}|>{\centering\arraybackslash}X|}
\hline
\textbf{Symbol} & \textbf{Meaning}\\
\hline
$\omega$ & Warping window band width used to constrain admissible matches. \\
\hline
$M(x_i,q_j)$ & Matching cost between $x_i$ and $q_j$. \\
\hline
$D(x_i)$,  & Deletion costs. \\
\hline
$G^{(X,Q)}$ & Induced bipartite graph for $(X,Q)$. \\
\hline
$w_e$ & Weight of edge $e$ in $G^{(X,Q)}$ induced by the corresponding operation cost. \\
\hline
$\{d_i\}_{i\in U\cup V}$ & Dual-feasible vertex weights satisfying Eq.~(\ref{equ_condition_DLP}). \\
\hline
$LB(\cdot,\cdot)$, $LB^*(\cdot,\cdot)$ & Lower bound and its symmetric version (max over swapped arguments). \\
\hline
\end{tabularx}
\end{table}

Given two time series \(X \in \mathbb{R}^n\) and \(Q \in \mathbb{R}^m\), along with an elastic similarity measure, we construct an undirected weighted graph  
\(G^{(X,Q)} = (U \cup V, E_1 \cup E_2)\), referred to as the induced graph of \(X\) and \(Q\).  
An illustration of this construction is shown in Figure~\ref{fig_The_Graph}.

For any two sequences $X \in \mathbb{R}^n$, $Q \in \mathbb{R}^m$ and an elastic similarity measure, 
we construct an undirected weighted graph $G^{(X,Q)}=(U \cup V, E_1\cup E_2)$, which is referred to as the induced graph of $X$ and $Q$.
Figure~\ref{fig_The_Graph} provides an illustration of the induced graph construction.

\begin{figure}[ht]
	\centering
    \includegraphics[width=0.65\linewidth]{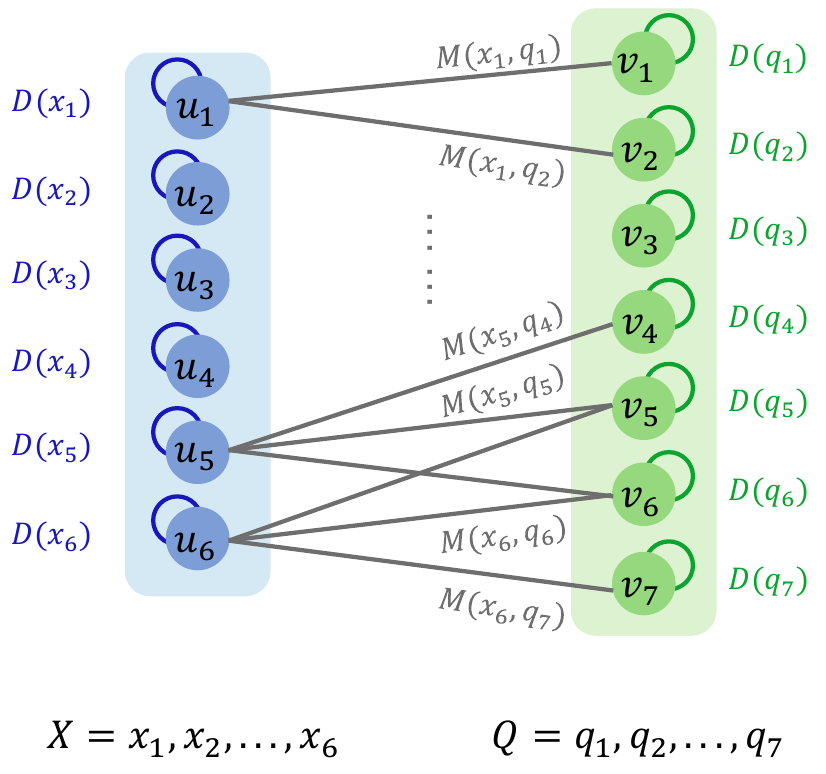}
	\caption{
        Illustration of the induced graph $G^{(X,Q)}$, which is a weighted bipartite graph with self-loops. 
        The weight of each edge in the graph is non-negative, with its specific value determined by both the distance measure and the values of the corresponding elements.
    }
	\label{fig_The_Graph}
\end{figure}

\noindent\textbf{Vertices (\(U\) and \(V\)).}  
We create a vertex \(u_i \in U\) for each element \(x_i \in X\), and a vertex \(v_j \in V\) for each element \(q_j \in Q\).

\noindent\textbf{Self-Loop Edges (\(E_2\))}.
For each \(u_i \in U\), we add a self-loop \((u_i, u_i)\) with weight \(w{(u_i,u_i)} = D_X(x_i) \geq 0\), representing the cost of deleting \(x_i\) from \(X\).  
Similarly, for each \(v_j \in V\), we add a self-loop \((v_j, v_j)\) with weight \(w{(v_j,v_j)} = D_Q(q_j) \geq 0\), corresponding to the deletion of \(q_j\) from \(Q\).  
The exact form of the deletion cost \(D(x_i)\) depends on the elastic measure. For instance:
In EDR, \(D(x_i) = 1\) for all \(x_i\);
In ERP, \(D(x_i) = (x_i - g)^2\), where \(g \in \mathbb{R}\) is a user-defined gap parameter.

\noindent\textbf{Cross Edges (\(E_1\))}.
We connect \(u_i \in U\) and \(v_j \in V\) with an undirected edge \((u_i, v_j)\) if \(|i - j| \leq \omega\), where $\omega$ is a warping window parameter.
The weight of $(u_i,v_j)$ is $w{(u_i,v_j)}=M(x_i,q_j)$, representing the matching cost between \(x_i\) and \(q_j\).
As with deletion costs, \(M(\cdot,\cdot)\) is defined by the specific elastic similarity measure.  
Some important notations are listed in Table \ref{fig_The_Graph}.

\subsection{Deriving Lower Bounds of Elastic Similarity Measures from the Induced Graph}
\label{section_mapping_lb_to_G}

Let \(ELD(X, Q)\) denote the abstract elastic distance between two time series \(X \in \mathbb{R}^n\) and \(Q \in \mathbb{R}^m\).  
Theorem~\ref{theroem_min_edge_cover} establishes that a valid lower bound on \(ELD(X, Q)\) can be obtained by computing the minimum-weight edge cover on the induced graph \(G^{(X,Q)}\).

\begin{theorem}
\label{theroem_min_edge_cover}
Given any \(X \in \mathbb{R}^n\) and \(Q \in \mathbb{R}^m\), let \(G^{(X,Q)} = (U \cup V, E_1 \cup E_2)\) be the induced graph constructed from \(X\) and \(Q\), and let \(S^* \subseteq E_1 \cup E_2\) be a minimum-weight edge cover of \(G^{(X,Q)}\).  
    Then the corresponding elastic similarity measure $ELD(X, Q)$ satisfies:
    
    \begin{equation}
    \label{equ_minimum_edge_cover}
        ELD(X,Q) \geq trans(\sum_{e\in S^*}w_e),
    \end{equation}
where $w_e$ is the weight of the edge $e \in  E_1 \cup E_2$ and $trans(\cdot)$ is the post transformation of the specific elastic measure listed in Table \ref{tab:elastic-measures-block}.
\end{theorem}

\begin{proof}
According to the definition of \(G^{(X,Q)}\),
a matching operation that matches $x_i$ and $q_j$ corresponds to the cross edge $(u_i,v_j) \in E_1$.
A deletion operation on element $x_i$ (or $q_j$) corresponds to the self-loop edge $(u_i,u_i) \in E_2$ (or $(v_j,v_j)$).

Thus, any valid operation sequence \(\mathbb{S}\) that transforms \(X\) into \(Q\) through match and delete operations maps to an edge set \(S \subseteq E_1 \cup E_2\), where each operation is represented by at most one edge.   
Moreover, the cost of any operation in $\mathbb{S}$ is always greater than or equal to the weight of its corresponding edge (if it exists) in \(G^{(X,Q)}\).  
Let $w_e$ be the non-negative weight of the edge $e$:

\begin{equation}
    RAC(n,m) \geq \sum_{e\in S}w_e.
\end{equation}

By Corollary~\ref{corollary_obeservation}, every element in \(X\) and \(Q\) must be matched or deleted in any valid alignment. Consequently, the corresponding edge set \(S\) must cover all vertices in \(U \cup V\), i.e., \(S\) is an edge cover of \(G\).  
Let $S^*$ be a minimum weight edge cover of $G^{(X,Q)}$, by definition,
\begin{equation}
    RAC(n,m) \geq \sum_{e\in S} w_e \geq   \sum_{e\in S^*}w_e.
\end{equation}
The proof is then completed by noting that the post-processing transformation \(\mathrm{trans}(\cdot)\) is strictly monotonic.
\end{proof}

Therefore, a valid lower bound on \(ELD(X, Q)\) is given by the cost of the minimum-weight edge cover on \(G^{(X,Q)}\).

To illustrate this correspondence concretely, consider the ERP distance with \(g = 1\)\footnote{\(g\) is the reference value, which is a parameter of ERP.}, applied to
\(X = (0, 5, 3, 7, 4)\) and \(Q = (6, 3, 6, 5, 2)\).
Table~\ref{table_example} lists an \emph{optimal} (minimum-cost) ERP alignment sequence and its associated costs.
Each operation maps directly to an edge in the induced graph: matches correspond to edges in \(E_1\), and deletions to self-loops in \(E_2\).
Critically, the resulting edge set covers all vertices in \(G^{(X,Q)}\).

Figure~\ref{fig:erp_overview}(b) visualizes this mapping, showing only the edges included in the edge cover.
This example demonstrates how the dissimilarity between \(X\) and \(Q\) under ERP can be interpreted as an edge cover problem on the induced graph.

\begin{table}[htbp]
\centering
\caption{Alignment Steps and Corresponding Edges for ERP\((X, Q)\)}
\label{table_example}
\scalebox{0.7}{
\begin{tabular}{|c|c|c|c|}
\hline
\textbf{Step} & \textbf{Operation} & \textbf{Cost} & \textbf{Corresponding Edge in \(G\)} \\ \hline
1 & delete \(x_1=0\) & \(D(x_1)=(0-g)^2=1\) & \((u_1,u_1) \in E_2\) \\ \hline
2 & match \(x_2=5\) with \(q_1=6\) & \(M(x_2,q_1)=(5-6)^2=1\) & \((u_2,v_1) \in E_1\) \\ \hline
3 & match \(x_3=3\) with \(q_2=3\) & \(M(x_3,q_2)=(3-3)^2=0\) & \((u_3,v_2) \in E_1\) \\ \hline
4 & match \(x_4=7\) with \(q_3=6\) & \(M(x_4,q_3)=(7-6)^2=1\) & \((u_4,v_3) \in E_1\) \\ \hline
5 & match \(x_5=4\) with \(q_4=5\) & \(M(x_5,q_4)=(4-5)^2=1\) & \((u_5,v_4) \in E_1\) \\ \hline
6 & delete \(q_5=2\) & \(D(q_5)=(2-g)^2=1\) & \((v_5,v_5) \in E_2\) \\ \hline
\end{tabular}
}
\end{table}

\subsection{A Relaxed Sufficient Condition for Valid Lower Bound}
\label{section_releasing_the_problem}

Computing the exact minimum-weight edge cover of the induced graph \(G^{(X,Q)}\) yields a  tight lower bound but is computationally expensive.  
Instead of solving for the cheapest edge cover on \(G^{(X,Q)}\), Theorem~\ref{the_graph} provides a relaxed and tractable sufficient condition stated as Equation~\eqref{equ_condition_DLP}:
We  assign each vertex \(i \in \{ U \cup V\}\) a non-negative “price tag” \(d_i \geq 0\) such that:
(i) \(d_i\) does not exceed the cost of deleting \(i\) (i.e., \(d_i \leq w{(i,i)}\)), and  
(ii) For every possible match \((u,v)\), the combined price satisfies \(d_u + d_v \leq w{(u,v)}\).

Intuitively, under these constraints, any edge cover must incur a total cost of at least \(\sum_i d_i\), since each vertex must be “paid for” by at least its assigned price \(d_i\). Thus, \(\sum_i d_i\) is a certified lower bound on the true alignment cost.
Formally described as follows:

\begin{theorem}
\label{the_graph}
Given time series \(X \in \mathbb{R}^n\) and \(Q \in \mathbb{R}^m\) with induced graph \(G^{(X,Q)} = (U \cup V, E_1 \cup E_2)\), let \(\{d_i\}_{i \in U \cup V}\) be a set of real numbers satisfying:
\begin{equation}
\label{equ_condition_DLP}
\begin{aligned}
d_u + d_v &\leq w{(u,v)} && \forall (u, v) \in E_1, \\
0 \leq d_i &\leq w{(i,i)} && \forall i \in U \cup V.
\end{aligned}
\end{equation}
Then,
\begin{equation}
ELD(X, Q) \geq trans\left(\sum_{i \in U \cup V} d_i\right).
\end{equation}
\end{theorem}

\begin{proof}
Let \(z_e \in \{0,1\}\) indicate whether the edge \(e \in E_1 \cup E_2\) is selected in the edge cover. 
The weighted minimum edge cover problem on the induced graph $G^{(X,Q)}$ can be formulated as the following integer program:

\begin{equation}
\label{equ_IP}
    \begin{aligned}
        Z_{IP}(G)=\min \quad & \sum_{e \in E_1\cup E_2} w_e z_e \\
        \text{s.t.} \quad & \sum_{e \in  E_1\cup E_2 }   \mathbf{1}_{\{ i \in e\}} \times z_e \geq 1, \quad \forall i \in U \cup V \\
        & z_e \in \{0,1\}, \quad \forall e \in E_1\cup E_2,  
    \end{aligned}    
\end{equation}
where  $\mathbf{1}_{\{ i \in e\}} =1$ \textbf{if and only if} the edge $e$ connects the vertex $i$, i.e. ,the vertex \(i\) is incident to the edge \(e\).

Relaxing the integrality constraint on $z_e$ yields the following linear program problem \(Z_{\mathrm{LP}}(G)\):

\begin{equation}
\label{equ_LP}
    \begin{aligned}
        Z_{LP}(G)=\min \quad & \sum_{e \in E_1\cup E_2} w_e z_e \\
        \text{s.t.} \quad & \sum_{e \in E_1\cup E_2} \mathbf{1}_{\{ i \in e\}} \times z_e \geq 1, \quad \forall i \in U \cup V \\
        & z_e \geq 0, \quad \forall e \in E_1\cup E_2.  
    \end{aligned}    
\end{equation}

Obviously, the feasible domain of variables $\{z_e\}_{e\in E_1 \cup E_2}$ in 
Equation (\ref{equ_LP}) strictly contains that of Equation (\ref{equ_IP}).
Therefore, the optimal solution in equation (\ref{equ_IP}) must be a feasible solution for equation (\ref{equ_LP}), which means, 
\begin{equation}
   \label{equ_realtion1}
   Z_{LP}(G) \leq Z_{IP}(G).
\end{equation}

The dual of \eqref{equ_LP} is to maximize the linear combination of the dual variables $\{d_i\}_{i \in  U \cup V}$, which is,

\begin{equation}
\label{equ_DLP}
    \begin{aligned}
    D_{LP}(G)=\max \quad & \sum_{i \in U \cup V} d_i \\
    \text{s.t.} \quad & d_u + d_v \leq w{(u,v)}, \quad & \forall (u, v) \in E_{1} \\
    & d_i \leq w{(i,i)}, \quad &\forall i \in U \cup V  \\
    & d_i \geq 0, \quad &\forall i \in U \cup V
    \end{aligned}    
\end{equation}

By weak duality, \(D_{\mathrm{LP}}(G) \leq Z_{\mathrm{LP}}(G)\).  
Moreover, from Theorem~\ref{theroem_min_edge_cover}, we know \(RAC(X, Q) \geq Z_{\mathrm{IP}}(G)\). 
Combining inequalities  (\ref{equ_minimum_edge_cover}) and (\ref{equ_realtion1}) yields:

\begin{equation}
     RAC(n,m)\geq   \sum_{e \in S^*} w_e = Z_{IP}(G) \geq Z_{LP}(G) \geq D_{LP}(G).
\end{equation}

Any assignment \(\{d_i\}\) satisfying \eqref{equ_condition_DLP} is a feasible solution for  
the dual linear programming problem described in equation (\ref{equ_DLP}).
It follows that,
\begin{equation}
    \sum_{i\in U \cup V} d_i \leq D_{LP}(G) \leq RAC(n,m).
\end{equation}

The proof is then completed by noting that \(\mathrm{trans}(\cdot)\) is strictly monotonic.
\end{proof}

\section{The Proposed Instantiation under the Paradigm: \textit{BGLB}}
\label{section_BGLB}

This section presents \textit{BGLB} (Bipartite Graph-based Lower Bound), a  instantiation of the proposed Bipartite Graph Edge-Cover Paradigm.  
Building upon the bipartite graph formulation introduced earlier, \textit{BGLB} constructs an efficient dual-feasible vertex weighting that satisfies the constraints derived from the linear programming relaxation of the minimum-weight edge cover problem.

\textit{BGLB} decomposes into three components: boundary, base, and augmentation, each capturing complementary aspects of sequence dissimilarity while ensuring full coverage of all elements in both time series in \textbf{\(O(n + m)\)} time. The resulting bound is provably correct and tighter than the state-of-the-art \textit{GLB} across all evaluated elastic measures, making \textit{BGLB} highly practical for large-scale similarity search.

\subsection{Formalization of \textit{BGLB}}
\label{4-1}  
We formally define \textit{BGLB} as an explicit construction of a dual-feasible vertex weighting on the induced graph \(G^{(X,Q)}\), satisfying the constraints in Theorem~\ref{the_graph}. This guaranties its validity as a lower bound while enabling linear-time computation.

For a time series \(X \in \mathbb{R}^n\), its upper and lower envelopes are defined as  
\[
\mathbb{U}^X_i = \max_{|k - i| \leq \omega} x_k, \quad
\mathbb{L}^X_i = \min_{|k - i| \leq \omega} x_k,
\]
for \(i = 1, \dots, n\), where \(\omega\) is the warping window radius and out-of-bound indices are clamped to the nearest valid position (e.g., \(x_k = x_1\) if \(k < 1\)).  
The envelopes \(\mathbb{U}^Q\) and \(\mathbb{L}^Q\) for \(Q\) are defined analogously.  
An illustration is shown in Figure~\ref{fig_envelope}.

\begin{figure}[ht]
	\centering
        \includegraphics[width=0.8\linewidth]{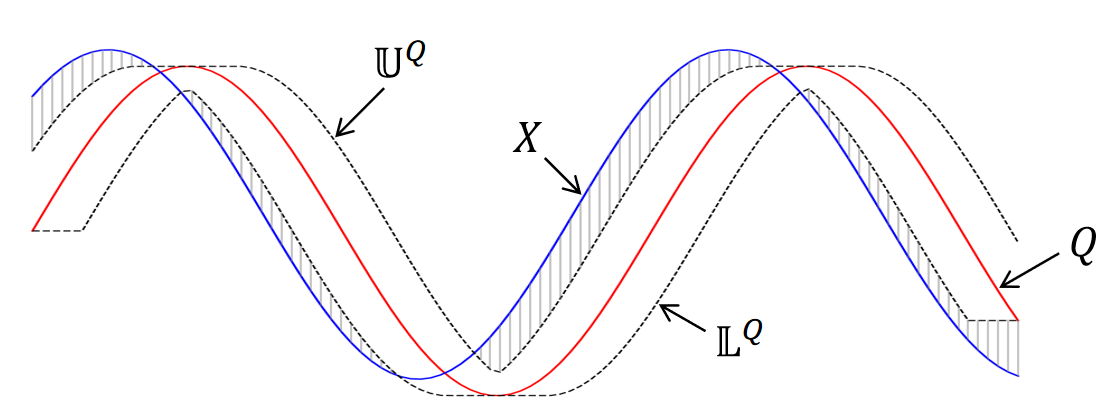}
        \caption{Illustration of the upper and lower envelopes.}
        \Description{Illustration of the upper and lower envelopes.}
        \label{fig_envelope}
\end{figure}

\begin{definition}(\textit{BGLB})
\label{def_BGLB}
For two sequences $X=(x_1,\ldots,x_n)\in\mathbb{R}^n$ and $Q=(q_1,\ldots,q_m)\in\mathbb{R}^m$, \textit{BGLB} is defined as:
\begin{equation}
\label{eq:BGLB_dir}
BGLB(X,Q)= trans\left( \Delta_{bdy}(X,Q)+\Delta_{base}(X,Q)+\Delta_{aug}(X,Q) \right),
\end{equation}
where $\Delta_{bdy}(X,Q)$, $\Delta_{base}(X,Q)$ and $\Delta_{aug}(X,Q)$ are the \textbf{Boundary Assignment term}, \textbf{Base Assignment term} and \textbf{Augmentation Assignment term}, respectively.

\smallskip
\noindent\textbf{Boundary Assignment term.}
$\Delta_{bdy}(X,Q)$ is the \emph{same boundary contribution as in \textit{GLB}}, which is measure-specific (it captures the mandatory corner constraints induced by the DP alignment) and is computed exactly in the same way as \textit{GLB}.

\smallskip
\noindent\textbf{Base Assignment term} is defined as:
\begin{equation}
\label{eq:BGLB_base}
\Delta_{base}(X,Q)=\sum_{i=2}^{n-1} d_{u_i},
\end{equation}
where

\[
d_{u_i} =
\begin{cases}
\delta(x_i,\mathbb{U}_i^Q), & x_i>\mathbb{U}_i^Q,\\
\delta(x_i,\mathbb{L}_i^Q), & x_i<\mathbb{L}_i^Q,\\
0, & \text{otherwise},
\end{cases}
\quad i=1,\ldots,n,
\]
and
\begin{equation}
\label{eq:def_delta}
\delta(x_i,b)=\min\{M(x_i,b),\,D(x_i)\}, \forall x_i, b \in \mathbb{R}.
\end{equation}

\smallskip
\noindent\textbf{Augmentation Assignment term} is defined as:

\begin{equation}
\label{eq:BGLB_aug}
\Delta_{aug}(X,Q)=\sum_{j=2}^{m-1}d_{v_j},
\end{equation}
where
\begin{equation}
\label{eq:BGLB_aug}
d_{v_j}=\left\{
\begin{aligned}
&\Gamma(q_j,\mathbb{U}_j^X),&q_j>\mathbb{U}_j^X\\
&\Gamma(q_j,\mathbb{L}_j^X),&q_j<\mathbb{L}_j^X\\
&0,&\text{otherwise}
\end{aligned}
\right.
\end{equation}
and 
\begin{equation}
\label{eq:def_gamma}
\Gamma(q_j,b)=\min\left\{\max\{M(q_j,b)-\mathbb{U}^{\Delta}_{j},0\},\,D(q_j)\right\}, \forall q_j,b \in \mathbb{R},
\end{equation}
\[
\mathbb{U}^{\Delta}_j = \max\{\, d_{u_i}\mid |i-j|\le \omega \,\}, \forall j\in\{1,\ldots,m\}.
\]
where out-of-range indices are ignored (or clamped) consistently with the envelope construction.
\end{definition}

Intuitively, \textbf{Boundary Assignment term} evaluates dissimilarity based on the first and last few elements of the two sequences, capturing potential mismatches at the temporal boundaries.
\textbf{Base Assignment term} computes the minimum cost of matching elements of \(Q\) to \(X\), leveraging envelope-based bounds.
\textbf{Augmentation Assignment term} computes the symmetric cost of matching elements of \(X\) to \(Q\).

\subsection{Correctness of \textit{BGLB}}
\label{section_proof_of_BGLB}

The correctness of \textit{BGLB} is established in Theorem~\ref{the_correct_lb}.
\begin{theorem}
\label{the_correct_lb}
Given any time series $X \in \mathbb{R}^n$  and $Q \in \mathbb{R}^m$, the following inequality holds:

\begin{equation}
    ELD(X,Q) \geq BGLB(X,Q).
\end{equation}
\end{theorem}

\begin{proof}
We first establish the inequality for Eq.~(\ref{eq:BGLB_dir}), and then obtain the symmetric bound by swapping the arguments.
Since all elastic measures considered in this paper are symmetric, we have $ELD(X,Q)=ELD(Q,X)$.

\smallskip
\noindent\textbf{Dual-feasible assignment.}
Recall from Definition~\ref{def_BGLB} that $d_{u_i}$ and $d_{v_j}$ are defined by the boundary-clipped values w.r.t.\ the envelopes
$\{\mathbb{L}_i^Q,\mathbb{U}_i^Q\}$ and $\{\mathbb{L}_j^X,\mathbb{U}_j^X\}$, together with the envelope
$\mathbb{U}^{\Delta}_j=\max\{y_{u_t}\mid |t-j|\le \omega\}$ used in $\Gamma(\cdot,\cdot)$.
By Definition~\ref{def_BGLB}, we also have
$\sum_i d_{u_i}=\Delta_{base}(X,Q)$ and $\sum_j d_{v_j}=\Delta_{aug}(X,Q)$.

We show that $\{d_{u_i}\}\cup\{d_{v_j}\}$ is dual-feasible for the dual LP in Eq.~(\ref{equ_condition_DLP}).

\smallskip
\noindent\textbf{Self-loop constraints ($E_2$).}
For any $u_i$, Definition~\ref{def_BGLB} gives $d_{u_i}\in\{0,\delta(x_i,\mathbb{U}_i^Q),\delta(x_i,\mathbb{L}_i^Q)\}$.
Since $\delta(a,b)=\min\{M(a,b),D(a)\}$ (Eq.~(\ref{eq:def_delta})), we have $0\le d_{u_i}\le D(x_i)=w(u_i,u_i)$.
For any $v_j$, by Eq.~(\ref{eq:def_gamma}) we have $0\le \Gamma(q_j,b)\le D(q_j)$, hence
$0\le d_{v_j}\le D(q_j)=w(v_j,v_j)$.
Therefore all constraints associated with $E_2$ hold.

\smallskip
\noindent\textbf{Match-edge constraints ($E_1$).}
Consider any $(u_i,v_j)\in E_1$. By definition, $|i-j|\le \omega$, which implies
$q_j\in[\mathbb{L}_i^Q,\mathbb{U}_i^Q]$ and $x_i\in[\mathbb{L}_j^X,\mathbb{U}_j^X]$.

We use the standard projection fact:
for any interval $[L,U]$ and any $a$, letting $b=\mathrm{clip}(a,[L,U])$ yields
$M(a,b)\le M(a,z)$ for all $z\in[L,U]$.

\medskip
\noindent\textbf{Case 1: $d_{v_j}=0$.}
Then $d_{u_i}+d_{v_j}=d_{u_i}$.
When $d_{u_i}=0$ the constraint is trivial.
Otherwise, $d_{u_i}=\delta(x_i,b_i)$ for $b_i\in\{\mathbb{L}_i^Q,\mathbb{U}_i^Q\}$, which is exactly $\mathrm{clip}(x_i,[\mathbb{L}_i^Q,\mathbb{U}_i^Q])$.
Thus
\[
d_{u_i}\le M(x_i,b_i)\le M(x_i,q_j)=w(u_i,v_j).
\]

\medskip
\noindent\textbf{Case 2: $d_{v_j}>0$.}
By Definition~\ref{def_BGLB} and the envelope definition of $\mathbb{U}^{\Delta}_j$, we have $d_{u_i}\le \mathbb{U}^{\Delta}_j$.
Moreover, $d_{v_j}>0$ implies that the inner $\max\{\cdot,0\}$ in Eq.~(\ref{eq:def_gamma}) is active, hence
\[
d_{v_j}
= \Gamma(q_j,b_j)
\le M(q_j,b_j)-\mathbb{U}^{\Delta}_j,
\]
where $b_j\in\{\mathbb{L}_j^X,\mathbb{U}_j^X\}$ equals $\mathrm{clip}(q_j,[\mathbb{L}_j^X,\mathbb{U}_j^X])$.
Therefore,
\[
\begin{aligned}
d_{u_i}+d_{v_j}
&\le \mathbb{U}^{\Delta}_j+\bigl(M(q_j,b_j)-\mathbb{U}^{\Delta}_j\bigr) \\
&= M(q_j,b_j) \\
&\le M(q_j,x_i)=M(x_i,q_j)=w(u_i,v_j).
\end{aligned}
\]

where the last inequality uses $x_i\in[\mathbb{L}_j^X,\mathbb{U}_j^X]$ and the projection fact.

\smallskip
Combining the above, all constraints for edges in $E_1$ and $E_2$ are satisfied, so the assignment is dual-feasible.
By Theorem~\ref{the_graph},
\[
\begin{aligned}
ELD(X,Q) \;&\ge\; trans\left( \Delta_{bdy}(X,Q)+\sum_i d_{u_i}+\sum_j d_{v_j} \right)\\
          &=\; trans\left( \Delta_{bdy}(X,Q)+\Delta_{base}(X,Q)+\Delta_{aug}(X,Q) \right) \\
          &=\; BGLB(X,Q).
\end{aligned}
\]

Applying the same argument to $(Q,X)$ and using symmetry of $ELD$, we further obtain
$ELD(X,Q)\ge \max\{BGLB(X,Q),BGLB(Q,X)\}=BGLB^*(X,Q)$.
\end{proof}

\subsection{An Illustrative Example of \textit{BGLB}}

Consider two time series \(X\) and \(Q\) with warping window size \(\omega\).
Figure~\ref{fig_D_LP} illustrates how the proposed \textit{BGLB} differs from base-only bounds such as \textit{GLB}: the base pass assigns weights to vertices in \(U\), while the augmentation pass assigns additional non-negative weights to vertices in \(V\) to ensure full coverage under the dual constraints.

\begin{figure}[ht]
    \centering
    \subfigure[ The base-only sum is \(\sum_i d_{u_i}=5\). With augmentation on \(V\), \(BGLB(X,Q)=\sum_i d_{u_i}+\sum_j d_{v_j}=13\).]{
        \includegraphics[width=0.45\linewidth]{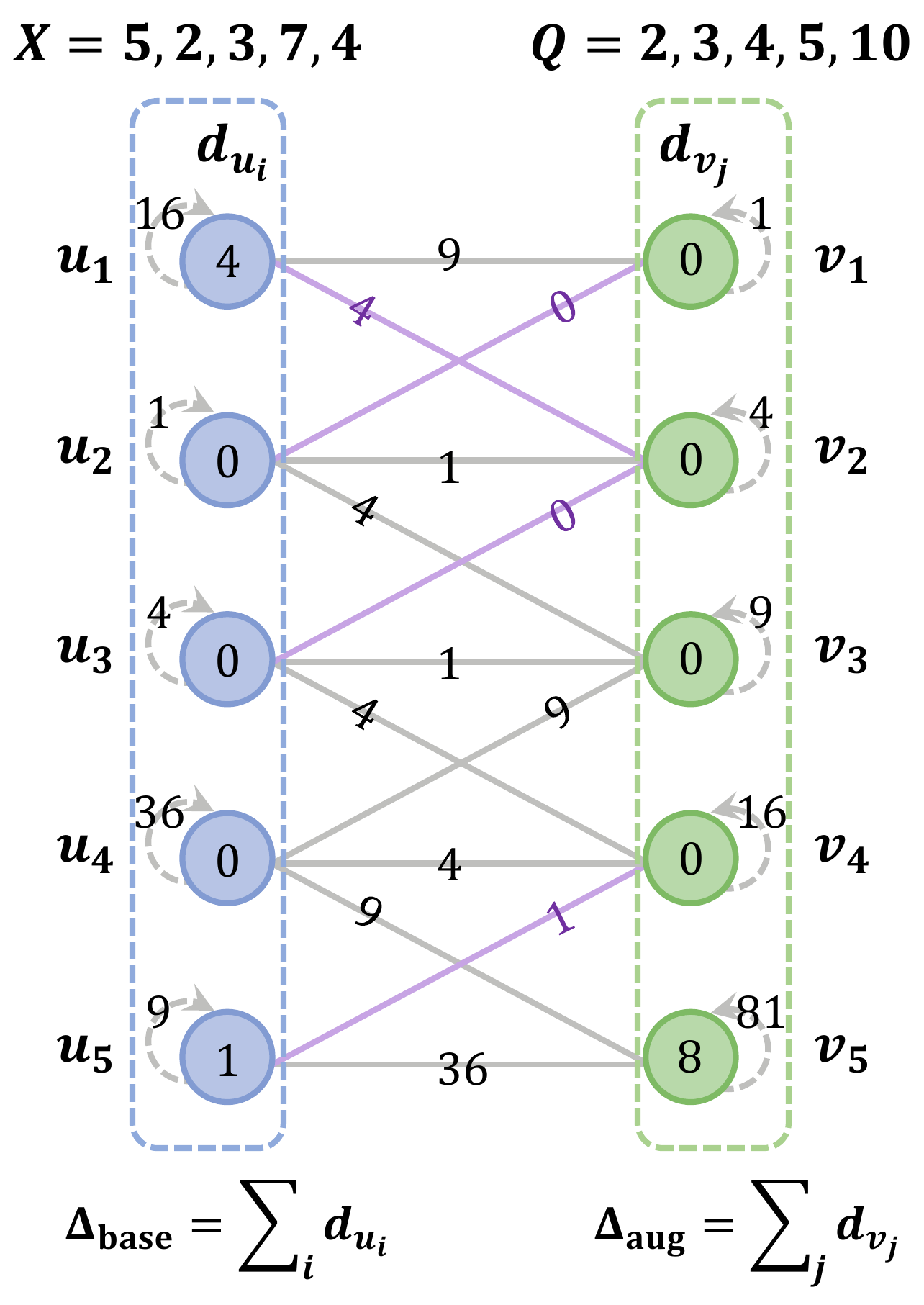}
    }
    \hfill
    \subfigure[The base-only sum is \(\sum_i d'_{u_i}=9\). With augmentation on \(V\), \(BGLB(Q,X)=\sum_i d'_{u_i}+\sum_j d'_{v_j}=13\).]{
        \includegraphics[width=0.45\linewidth]{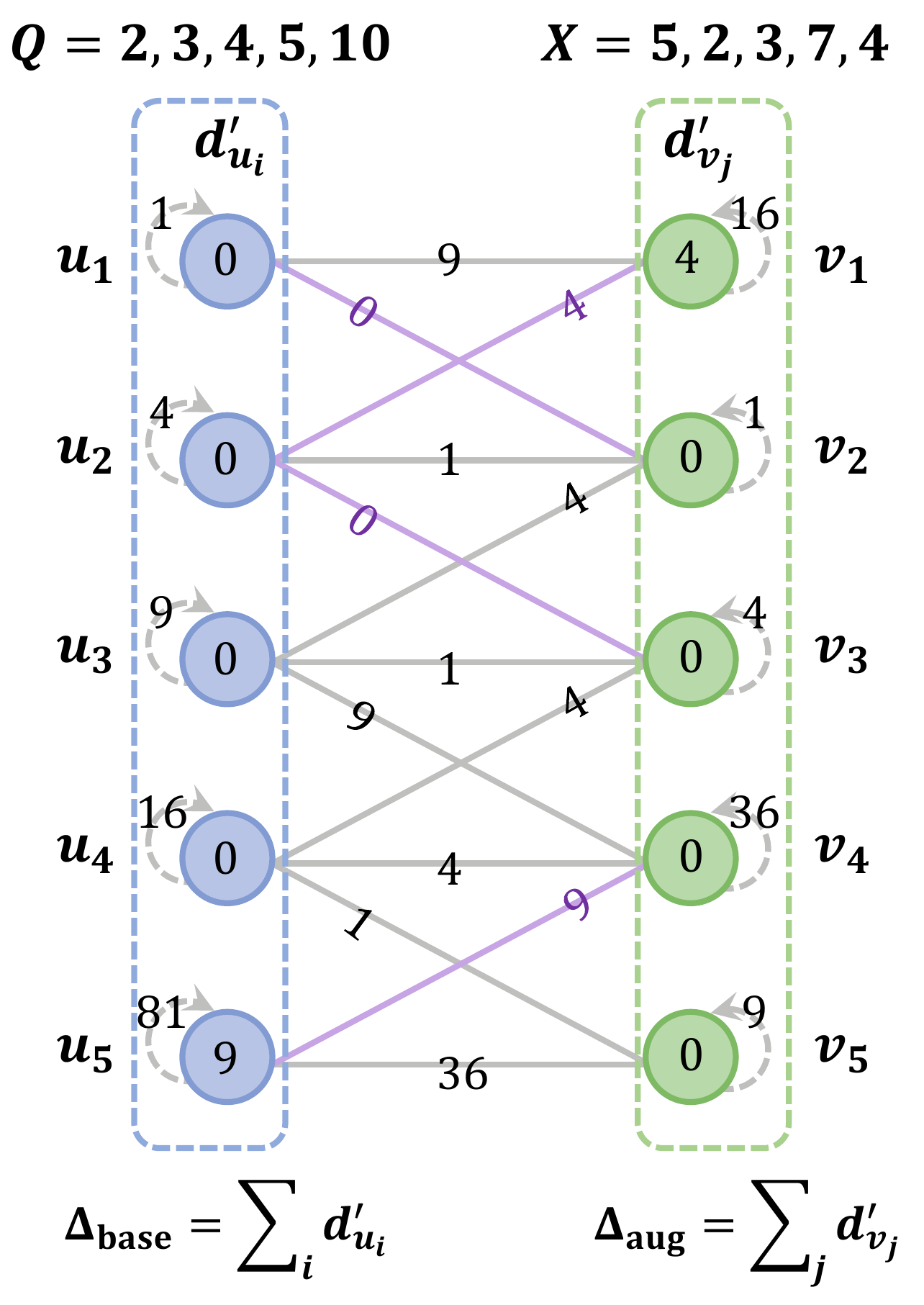}
    }
    \caption{An ERP example (\(\omega=1, g=1\)) illustrating why \textit{BGLB} is tighter than the lower bound obtained from the existing strip-based paradigm. For clarity, the boundary term is omitted in this visualization, focusing on the base and augmentation terms to show the symmetric case.}

    \label{fig_D_LP}
\end{figure}

We take \textit{ERP} distance as an example, considering \(X = (5, 2, 3, 7, 4)\) and \(Q = (2, 3, 4, 5, 10)\) with warping window \(\omega = 1\) and reference value \(g = 1\).
Table~\ref{tab:selected_series} reports the deletion costs and the vertex weights produced by the algorithm in both directions.
For \textit{GLB} (base-only), the bound is the sum on \(U\): \(GLB(X,Q)=\sum_i d_{u_i}=5\) and \(GLB(Q,X)=\sum_i d'_{u_i}=9\), so the symmetrized bound is \(GLB^{\ast}(X,Q)=\max\{5,9\}=9\).
For \textit{BGLB}, we additionally include the augmentation weights on \(V\): \(\sum_j d_{v_j}=8\), giving \(BGLB(X,Q)=5+8=13\).
Computing this augmentation pass for \((X,Q)\) also produces the base weights \(\{d'_{u_i}\}\) for \((Q,X)\); we then only need one additional augmentation pass to obtain \(\{d'_{v_j}\}\) and \(BGLB(Q,X)\).
In this example, \(BGLB(Q,X)=\sum_i d'_{u_i}+\sum_j d'_{v_j}=13\), hence \(BGLB^{\ast}(X,Q)=13\), which remains larger than \(GLB^{\ast}(X,Q)=9\).

\begin{table}[h]
    \centering
    \caption{Variables involved in an example using \textit{BGLB}}
    \scalebox{0.8}{ 
    \begin{tabular}{|c|c|c|c|c|c|c|c|c|c|c|}
        \hline
        Index & $X$ & $Q$ & $D(x_i)$ & $D(q_j)$ & $d_{u_i}$ & $d_{u_i}'$ & $d_{v_j}$ & $d_{v_j}'$ & $d_{u_i}+d_{v_j}$ & $d_{u_i}'+d_{v_j}'$ \\
        \hline
        1 & 5 & 2 & 16 & 1 & 4 & 0 & 0 & 4 & 4 & 4 \\
        \hline
        2 & 2 & 3 & 1 & 4 & 0 & 0 & 0 & 0 & 0 & 0 \\
        \hline
        3 & 3 & 4 & 4 & 9 & 0 & 0 & 0 & 0 & 0 & 0 \\
        \hline
        4 & 7 & 5 & 36 & 16 & 0 & 0 & 0 & 0 & 0 & 0 \\
        \hline
        5 & 4 & 10 & 9 & 81 & 1 & 9 & 8 & 0 & 9 & 9 \\
        \hline
    \end{tabular}
    }
    \label{tab:selected_series}
\end{table}



Importantly, this framework is agnostic to the specific elastic measure or the underlying matching/deletion cost definitions. 
It only requires that the measure conform to the basic structure of matching and deletion operations with non-negative costs. 
This generality makes \textit{BGLB} a flexible and extensible computational paradigm.

\subsubsection{The superiority of \textit{BGLB}.}
Typical lower bounds for elastic measures only consider $\Delta_{bdy}(X,Q)$ and $\Delta_{base}(X,Q)$ in \textit{BGLB}, without taking into account the additional lower bounds provided by the augmentation computation.
By comparing the mathematical forms of the previous work \textit{GLB}, the proposed \textit{BGLB} is better than \textit{GLB} in terms of lower bound effectiveness.

The symmetric version of $BGLB(X,Q)$, denoted as $BGLB^*$:
\begin{equation}
\label{eq:BGLB_sym}
\begin{split}
BGLB^*(X,Q) & = \max\!\left\{BGLB(X,Q),\,BGLB(Q,X)\right\}.
\end{split}
\end{equation}
Unless explicitly stated otherwise, throughout the paper (including all experiments) we use the symmetric form for both \textit{GLB} and \textit{BGLB}, which are $GLB^*$ and $BGLB^*$.

\begin{corollary}[\textit{BGLB} is always tighter than \textit{GLB}]
\label{cor:BGLB_ge_GLB}
Given any $X \in \mathbb{R}^n$ and $Q \in \mathbb{R}^m$, the following inequality holds for every elastic similarity measure:
\begin{equation}
    BGLB^*(X,Q) \geq GLB^*(X,Q).
\end{equation}
\end{corollary}

\begin{proof}
Write the \textit{GLB} as
\begin{equation}
\label{eq:GLB_dir}
GLB(X,Q)=\Delta_{bdy}(X,Q)+\Delta_{base}(X,Q).
\end{equation}
By Definition~\ref{def_BGLB}, we have
$BGLB(X,Q)=GLB(X,Q)+\Delta_{aug}(X,Q)\ge GLB(X,Q)$
because $\Delta_{aug}(X,Q)\ge 0$.
Finally, recall that for any lower bound $LB$, the symmetric version is defined as
\[
LB^{*}(X,Q)=\max\{LB(X,Q),\;LB(Q,X)\}.
\]
Since the inequality holds for both $LB(X,Q)$ and $LB(Q,X)$, it holds for their maximums.
Therefore, $BGLB^*(X,Q) \geq GLB^*(X,Q)$.
\end{proof}

\subsection{Computing Lower Bounds with \textit{BGLB}}

\subsubsection{Computing \textit{BGLB} with Early Stopping Methodology}

The goal of similarity search is to identify the time series pairs that are sufficiently similar to each other. 
Instead, it suffices to obtain a lower bound that is tight enough to enable effective pruning.
By definition, \textit{BGLB} is the sum of three components: \textbf{Boundary term},\textbf{Base matching cost} and \textbf{Augmentation matching cost}.
\textit{BGLB} also adheres to the early stopping principle.
The full bound is the sum of these components. 
In practice, we compute them incrementally and apply \emph{early stopping}: if the partial sum exceeds the current best-so-far (bsf) distance, we terminate immediately. Algorithm~\ref{ALGORITHM_BGLB} details this procedure.

\begin{algorithm}[ht]
    \caption{$BGLB(X, Q)$}
    \footnotesize
    \label{ALGORITHM_BGLB}
    \LinesNumbered
    \KwIn{\textbf{Time series $X\in \mathbb{R}^{n}$, $Q \in \mathbb{R}^{m}$; best-so-far value $\textit{bsf}$}}
    \KwOut{\textbf{A lower bound for the elastic distance}}
    $\Delta_{bdy} \leftarrow$ Dist computed from the boundary  \; 
    $\Delta_{base1} \leftarrow 0, \Delta_{aug1} \leftarrow 0, \Delta_{base2} \leftarrow 0, \Delta_{aug2} \leftarrow 0$\;

    \For{$i \leftarrow 2 $ \KwTo $n-1$}{
        \If{$x_i>\mathbb{U}_i^Q$}{
            $d_i \leftarrow \delta(x_i,\mathbb{U}_i^Q)$\; 
        }
        \ElseIf{$x_i<\mathbb{L}_i^Q$}{
            $d_i \leftarrow \delta(x_i,\mathbb{L}_i^Q)$\; 
        }
        \Else{ 
           $d_i \leftarrow 0$\;
        }
        $\Delta_{base1} \leftarrow \Delta_{base1}+ d_i$\;
        \If{$trans\left(\Delta_{bdy}+\Delta_{base1}\right) > \textit{bsf}$ }{
            \Return $\infty$ \tcp*{early stopping}
        }
    }
    $\mathbb{U}^{\Delta1}\leftarrow$ upper envelope (windowed maxima under $\omega$) induced by $\{d_i\}$\;
    
    \For{$j \leftarrow 2$ \KwTo $m-1$}{
        \If{$q_j>\mathbb{U}_j^X$}{
            $d_j' \leftarrow \delta(q_j,\mathbb{U}_j^X)$; $\gamma_j \leftarrow \Gamma(q_j,\mathbb{U}_j^X)$\;
        }
        \ElseIf{$q_j<\mathbb{L}_j^X$}{
            $d_j' \leftarrow \delta(q_j,\mathbb{L}_j^X)$; $\gamma_j \leftarrow \Gamma(q_j,\mathbb{L}_j^X)$\;
        }
        \Else{
           $d_j' \leftarrow 0$; $\gamma_j \leftarrow 0$\;
        }
    
        $\Delta_{aug1} \leftarrow \Delta_{aug1}+ \gamma_j$\;
        $\Delta_{base2} \leftarrow \Delta_{base2}+ d_j'$\;
        
        \If{$(trans\left(\Delta_{bdy}+\Delta_{base1}+\Delta_{aug1}\right) > \textit{bsf})$ \textbf{or} $(trans\left(\Delta_{bdy}+\Delta_{base2}\right) > \textit{bsf})$}{
            \Return $\infty$ \tcp*{early stopping}
        }
    }
 
    $\mathbb{U}^{\Delta2}\leftarrow$ upper envelope (windowed maxima under $\omega$) induced by $\{d'_j\}$\;

    \For{$i \leftarrow 2$ \KwTo $n-1$}{
        \If{$x_i>\mathbb{U}_i^Q$}{
            $\gamma_i' \leftarrow \Gamma(x_i,\mathbb{U}_i^Q)$\;
        }
        \ElseIf{$x_i<\mathbb{L}_i^Q$}{
            $\gamma_i' \leftarrow \Gamma(x_i,\mathbb{L}_i^Q)$\; 
        }
        \Else{
           $\gamma_i' \leftarrow 0$\; 
        }
        $\Delta_{aug2} \leftarrow \Delta_{aug2}+ \gamma_i'$\;
        \If{ $trans\left(\Delta_{bdy}+\Delta_{base2}+\Delta_{aug2}\right) > \textit{bsf}$ }{
            \Return $\infty$ \tcp*{early stopping}
        }
    }
    
    \Return $trans\left(\Delta_{bdy}+ \max\{\Delta_{base1}+\Delta_{aug1},\Delta_{base2}+\Delta_{aug2}\}\right)$\;
\end{algorithm}

\subsubsection{Algorithm Analysis} 
The upper and lower envelopes can be calculated in $O(m)$ or $O(n)$ time using priority queues, and the time required to traverse each element in $X$ and $Q$ during base and augmentation computations is only constant time.
Therefore, the total time complexity is $O(m+n)$, and the space complexity can also be calculated using a similar logic, which is $O(m+n)$.

\subsubsection{Implementation note (discrete measures).}
For discrete measures with bounded matching costs (e.g., TWED, LCSS, EDR, SWALE), the \textit{DBGLB} variant optimizes the augmentation term by avoiding the need for dynamic envelope computation or maintenance. Instead, \textit{DBGLB} employs a greedy two-pointer approach that efficiently cancels assignments within the allowable warping window, preserving the same bound value and early-stopping criteria.

In the base pass, \textit{DBGLB} computes the base vertex weights for both directions (from \(X\) to \(Q\) and from \(Q\) to \(X\)) using an envelope-based approach similar to \textit{BGLB}. These computed masses represent the minimum coverage demand that must be fulfilled within the warping window. This process is performed once per direction and the results are stored.

During the augmented assignment term, \textit{DBGLB} scans from the earliest uncovered position, using the available mass to satisfy demand. Residual demands that cannot be covered are added to the augmentation term \(\Delta_{\text{aug}}\). The entire process operates in \(O(n+m)\) time, ensuring each element is processed only once.

This implementation also supports early stopping: once the base and boundary terms are computed, and during the augmentation pass, \textit{DBGLB} checks if the partial lower bound exceeds the best-so-far threshold. If so, the algorithm halts early, saving computation time.

By eliminating the need for dynamic envelope computation and relying solely on sequential scans, \textit{DBGLB} is particularly effective for large-scale similarity search tasks where computational efficiency is critical. \footnote{Specific implementation details can be found at \url{https://github.com/BoyuXiao/BGLB/blob/main/Appendix.pdf}.}

\section{Experiments}
\label{sec:exp}

\subsection{Experimental Setup}
\label{sec:exp_setup}

\noindent\textbf{Hardware Environment.}
All experiments are conducted on a machine equipped with an Intel i7-10700 CPU, 32GB RAM and 1TB HDD, running Windows 11.

\noindent\textbf{Dataset.}
All evaluations are conducted on all 128 datasets of UCR time series classification archive~\cite{UCRArchive}.
Each dataset in this archive is pre-partitioned into a training set and a test set.
This benchmark is the most widely adopted in the time series community.\footnote{Following standard practice, we use the equal-length version of the UCR dataset.}

\noindent\textbf{Competitors.}
We compare \textit{BGLB} against 18 existing lower bounds across seven elastic measures, as summarized in Table~\ref{table_LBS}.
Our focus is on general-purpose elastic measures (e.g., in \textit{DTW}-based nearest neighbor search).
Thus, we exclude lower bounds that are specially designed for subsequence matching or self-join tasks~\cite{DBLP:journals/pvldb/ChaoZQW25,DBLP:journals/tkde/ChaoGMLW25}.
All experiments are implemented in C/C++ with \texttt{-O2} compiler optimization, and our code is available at \url{https://github.com/BoyuXiao/BGLB}.

\noindent\textbf{Parameters.}
Unless otherwise specified, we use the standard parameter settings summarized in Table~\ref{table_LBS}.
We set the warping window to $\omega=\mathrm{round}(0.05\,n)$, i.e., 5\% of the sequence length.

\begin{table}[htbp]
\centering
\caption{Comparative lower bounds and default parameters.}
\fontsize{5}{8}\selectfont{
\begin{tabular}{|c|c|c|}
\hline
\textbf{Elastic Measure} & \textbf{Comparative Lower Bounds} & \textbf{Parameters} \\
\hline

\multirow{2}{*}{\textit{DTW}}
& LB\_Kim~\cite{914875}, LB\_Keogh~\cite{10.5555/2993953.2994027,KEOGH2002406}, LB\_Improved~\cite{LEMIRE20092169}
& \multirow{2}{*}{--} \\
& LB\_Petitjean~\cite{WEBB2021107895}, LB\_Webb~\cite{WEBB2021107895},
GLB\_\textit{DTW}~\cite{10.14778/3594512.3594530} 
& \\
\hline

\multirow{2}{*}{\textit{ERP}}
& LB\_\textit{ERP}~\cite{10.5555/1316689.1316758}, LB\_Kim-\textit{ERP}~\cite{10.5555/1316689.1316758}, GLB\_\textit{ERP}~\cite{10.14778/3594512.3594530},
& \multirow{2}{*}{$g=0$} \\
& LB\_Keogh-\textit{ERP}~\cite{10.5555/1316689.1316758}
& \\
\hline

\multirow{1}{*}{\textit{MSM}}
& LB\_\textit{MSM}~\cite{Tan2020FastEE}, GLB\_\textit{MSM}~\cite{10.14778/3594512.3594530}
& $c=0.5$ \\
\hline

\multirow{1}{*}{\textit{TWED}}
& LB\_\textit{TWED}~\cite{Tan2020FastEE}, GLB\_\textit{TWED}~\cite{10.14778/3594512.3594530}
& $\lambda=1,\; \nu=0.0001$ \\
\hline

\multirow{1}{*}{\textit{LCSS}}
& LB\_\textit{LCSS}~\cite{10.1145/956750.956777,Vlachos2006Indexing,994784}, GLB\_\textit{LCSS}~\cite{10.14778/3594512.3594530}
& $\epsilon=0.2$ \\
\hline

\multirow{1}{*}{\textit{EDR}}
& GLB\_\textit{EDR}~\cite{10.14778/3594512.3594530}
& $\epsilon=0.1$ \\
\hline

\multirow{1}{*}{\textit{SWALE}}
& GLB\_\textit{SWALE}~\cite{10.14778/3594512.3594530}
& $\epsilon=0.2,\; p=5,\; r=1$ \\
\hline
\end{tabular}
}
\label{table_LBS}
\end{table}


\begin{figure*}[ht!]
    \centering
    \includegraphics[width=1.0\linewidth]{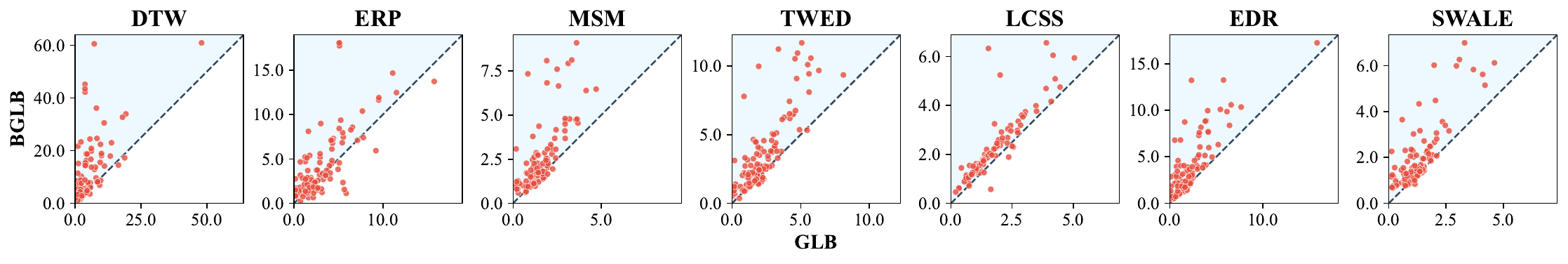}
    \vspace{-0.4cm} 
    \caption{Comparison of Speedup for different lower bounds on the 128 UCR datasets (1-NN search).}
    \label{fig_speedup}
\end{figure*}

\begin{figure*}[ht!]
    \centering
    \includegraphics[width=1.0\linewidth]{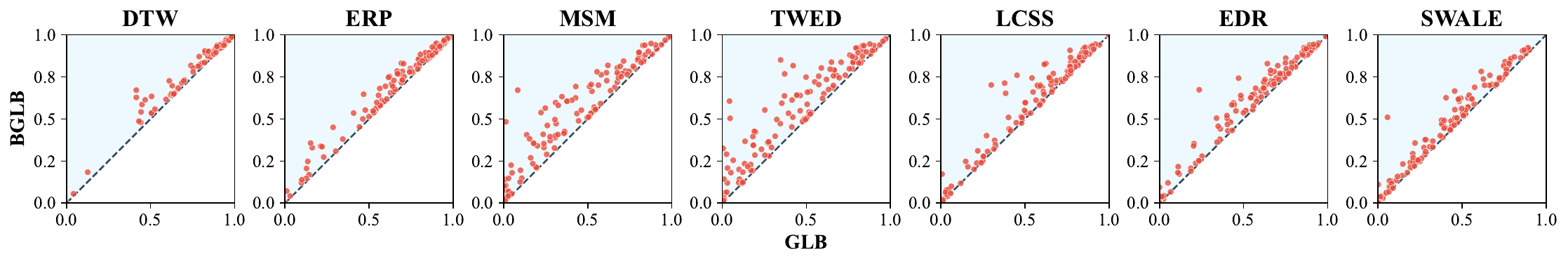}
    \vspace{-0.4cm}
    \caption{Comparison of Pruning rate between \textit{BGLB} and \textit{GLB} on the 128 UCR datasets.}
    \label{fig_pruning_ratio}
\end{figure*}

\begin{figure*}[ht!]
    \centering
    \includegraphics[width=1.0\linewidth]{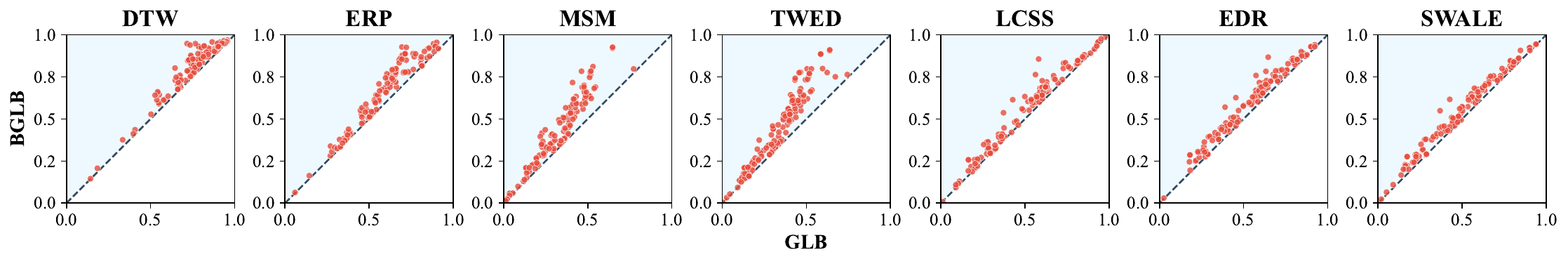}
    \vspace{-0.4cm}
    \caption{Comparison of Average TLB for different lower bounds on the 128 UCR datasets.}
    \label{fig_TLB}
\end{figure*}

\subsection{Speedup Evaluation on 1-NN Similarity Search}
\label{sec:exp_1nn_speed}

We first evaluate the impact of lower bounds on end-to-end exact similarity search using the 1-Nearest Neighbor (1-NN) task, which is a standard workload in time series classification.
In our experiments, each sequence in the test set is used as a query \(Q\), and we perform a 1-NN search against all time series in the training set.
The procedure follows the classic lower-bound filtering pipeline: Since pruning uses a valid lower bound, the returned nearest neighbor (and thus the 1-NN classification label) is identical to the naive exact baseline; we only reduce the amount of computation.
Because our pruning is exact, downstream accuracy is unchanged; we therefore report runtime and pruning effects.
For each candidate \(X\), we compute a lower bound \(LB(Q,X)\), only when \(LB(Q,X) < \textit{bsf}\) (best-so-far) do we compute the exact elastic distance \(ELD(Q,X)\) and potentially update \(\textit{bsf}\).

We report the average speedup compared with the naive 1-NN approach without lower bounds.
Figure~\ref{fig_speedup} summarizes the results over 128 datasets.
Across all measures, \textit{BGLB} consistently improves end-to-end performance over \textit{GLB} (the current SOTA general lower bound),
with average speedups of $84.9\%$ for \textit{DTW}, $24.6\%$ for \textit{ERP}, $37.4\%$ for \textit{MSM}, $48.2\%$ for \textit{TWED}, $63.5\%$ for \textit{EDR}, $39.2\%$ for \textit{SWALE} and $30.4\%$ for \textit{LCSS}.
These gains indicate that tighter bounds translate into fewer expensive exact-distance evaluations and thus faster exact search.

\subsection{Evaluation on Pruning Ratio}
\label{sec:exp_pruning_ratio}

To directly expose the mechanism behind the speedups, we measure the pruning ratio, i.e. the fraction of candidate series that are safely discarded by the lower bound before invoking the exact DP.
Figure~\ref{fig_pruning_ratio} compares \textit{BGLB} and \textit{GLB}.
Across all elastic measures, \textit{BGLB} consistently prunes a larger portion of candidates, explaining the runtime improvements observed in Figure~\ref{fig_speedup}.
This pattern is particularly pronounced on measures where strong DTW-specific bounds do not exist (e.g., MSM/TWED/EDR/SWALE), highlighting the value of a general framework that systematically strengthens pruning power.

\subsection{Tightness of Lower Bounds (TLB)}
\label{sec:exp_tlb}

The pruning behavior above is driven by how tightly a lower bound approximates the exact distance.
Following prior work, we quantify this quality using Tightness of Lower Bound (TLB)~\cite{keogh2001dimensionality},
defined as the average ratio between a lower bound and the exact elastic distance \(ELD(X,Q)\).
To account for potential asymmetry, we use
\(
LB^{*}(X,Q)=\max\{LB(X,Q),\,LB(Q,X)\}.
\)
For each dataset, we set the training set as $\mathbb{X}$ and the test set as $\mathbb{Q}$ and compute
\begin{equation}
\text{TLB} = \mathop{Average}\limits_{X \in \mathbb{X},\, Q \in \mathbb{Q}}
\left(
\frac{LB^{*}(X,Q)}{\text{Real Elastic Distance between }X \text{ and } Q}
\right).
\end{equation}
A larger TLB indicates a tighter bound and, in turn, stronger pruning potential.
By definition, a valid lower bound must satisfy \(LB^{*}(X,Q) \le ELD(X,Q)\) for every pair \((X,Q)\).

Table~\ref{table_tlb_stats} reports TLB statistics across all 128 datasets, summarized by average, minimum, maximum, and quartiles.
Figure~\ref{fig_TLB} further compares \textit{BGLB} and \textit{GLB} on each dataset; points above the diagonal correspond to datasets on which \textit{BGLB} is tighter.
Overall, \textit{BGLB} achieves substantially higher tightness than existing general-purpose lower bounds across all measures; for \textit{DTW}, specialized DTW-specific bounds can be tighter, while \textit{BGLB} remains competitive and provides a favorable tightness--efficiency trade-off (Figure~\ref{fig_the_lower_dound_of_Elatic_distance_G}).

\begin{table}[htbp]
    \centering
    \caption{Tightness of various lower bounds for different elastic similarity measures, evaluated over 128 UCR datasets.}
    \label{table_tlb_stats}
    \scalebox{0.7}{
    \begin{tabular}{llcccccc}
        \toprule
        \textbf{Elastic} & \textbf{Lower} & \multirow{2}{*}{\textbf{Average}} & \multirow{2}{*}{\textbf{Min}} & \multirow{2}{*}{\textbf{Q1}} & \multirow{2}{*}{\textbf{Median}} & \multirow{2}{*}{\textbf{Q3}} & \multirow{2}{*}{\textbf{Max}} \\
        \textbf{Measure} & \textbf{Bounds} & & & & & & \\
        \midrule
        \multirow{6}{*}{\textit{DTW}} 
        & LB\_Kim$^*$ & 0.1668 & 0.0012 & 0.0564 & 0.1106 & 0.2260 & 0.7227 \\
        & LB\_Keogh$^*$ & 0.7470 & 0.0432 & 0.6740 & 0.7593 & 0.8505 & 0.9563 \\
        & LB\_Improved$^*$ & 0.8152 & 0.0514 & 0.7678 & 0.8606 & 0.9105 & 0.9696 \\
        & \underline{LB\_Webb}$^*$ & \underline{0.8342} & \underline{0.1743} & \underline{0.7786} & \underline{0.8834}$^*$ & \underline{0.9312} & \underline{0.9792} \\
        & \textbf{LB\_Petitjean}$^*$ & \textbf{0.8394} & \textbf{0.1743} & \textbf{0.7791} & \textbf{0.8913} & \textbf{0.9378} & \textbf{0.9797} \\
        & GLB\_\textit{DTW}$^*$ & 0.7574 & 0.1431 & 0.6778 & 0.7761 & 0.8513 & 0.9566 \\
        & BGLB\_\textit{DTW}$^*$ & 0.8131 & 0.1451 & 0.7486 & 0.854 & 0.9211 & 0.9688 \\
        \midrule
        \multirow{5}{*}{\textit{ERP}} 
        & LB\_\textit{ERP}$^*$ & 0.0210 & 0.0000 & 0.0000 & 0.0001 & 0.0004 & 0.5946 \\
        & LB\_Kim-\textit{ERP}$^*$ & 0.1360 & 0.0012 & 0.0520 & 0.0937 & 0.1865 & 0.6133 \\
        & LB\_Keogh-\textit{ERP}$^*$ & 0.5723 & 0.0435 & 0.4459 & 0.6051 & 0.6941 & 0.9096 \\
        & GLB\_\textit{ERP}$^*$ & \underline{0.5926} & \underline{0.0591} & \underline{0.4955} & \underline{0.6178} & \underline{0.6970} & \underline{0.9115} \\
        & \textbf{BGLB\_\textit{ERP}}$^*$ & \textbf{0.6675} & \textbf{0.0633} & \textbf{0.5258} & \textbf{0.7155} & \textbf{0.8175} & \textbf{0.9546} \\
        \midrule
        \multirow{3}{*}{\textit{MSM}} 
        & LB\_\textit{MSM}$^*$ & 0.1758 & 0.0038 & 0.1111 & 0.1658 & 0.2428 & 0.3956 \\
        & GLB\_\textit{MSM}$^*$ & \underline{0.3226} & \underline{0.0082} & \underline{0.2255} & \underline{0.3509} & \underline{0.4207} & \underline{0.7730} \\
        & \textbf{BGLB\_\textit{MSM}}$^*$ & \textbf{0.4267} & \textbf{0.0124} & \textbf{0.2918} & \textbf{0.4108} & \textbf{0.5749} & \textbf{0.9265} \\
        \midrule
        \multirow{3}{*}{\textit{TWED}} 
        & LB\_\textit{TWED}$^*$ & 0.0095 & 0.0000 & 0.0006 & 0.0014 & 0.0031 & 0.5980 \\
        & GLB\_\textit{TWED}$^*$ & \underline{0.3454} & \underline{0.0039} & \underline{0.2249} & \underline{0.3643} & \underline{0.4510} & \underline{0.7434} \\
        & \textbf{BGLB\_\textit{TWED}}$^*$ & \textbf{0.4590} &\textbf{ 0.0040} & \textbf{0.2859} & \textbf{0.4740} & \textbf{0.6149} & \textbf{0.9107} \\
        \midrule
        \multirow{3}{*}{\textit{LCSS}} 
        & LB\_\textit{LCSS}$^*$ & 0.4738 & 0.0000 & 0.2791 &0.5357 & 0.6255 & 0.9274 \\
        & \underline{GLB\_\textit{LCSS}}$^*$ & \underline{0.4885} & \underline{0.0000} & \underline{0.2846} & \underline{0.5424} & \underline{0.6298} & \underline{0.9764} \\
        & \textbf{BGLB\_\textit{LCSS}}$^*$ & \textbf{0.5339} & \textbf{0.0000} & \textbf{0.3217} & \textbf{0.6050} & \textbf{0.7043} & \textbf{0.9849} \\
        \midrule
        \multirow{2}{*}{\textit{EDR}} 
        & {GLB\_\textit{EDR}}$^*$ & {0.5343} & {0.0000} &{0.3702} & {0.5732} & {0.6770} & {0.9234} \\
        & \textbf{BGLB\_\textit{EDR}}$^*$ & \textbf{0.5910} & \textbf{0.0000} & \textbf{0.4320} & \textbf{0.6341} & \textbf{0.7524} & \textbf{0.9407} \\
        \midrule
        \multirow{2}{*}{\textit{SWALE}} 
        & {GLB\_\textit{SWALE}}$^*$ & {0.5046} & {0.0189} & {0.3483} &{0.5006} & {0.6575} & {0.9392} \\
        & \textbf{BGLB\_\textit{SWALE}}$^*$ & \textbf{0.5529} & \textbf{0.0231} & \textbf{0.4083} & \textbf{0.5704} & \textbf{0.7111} & \textbf{0.9444} \\
        \bottomrule
    \end{tabular}
    }
\end{table}

\subsection{Complementarity with Pruned DP}
\label{sec:exp_eapruned}

Lower bounds and pruned dynamic programming address distinct bottlenecks in exact elastic similarity search.  
A lower bound performs \emph{candidate-level} filtering by discarding entire candidates before any dynamic programming computation begins,  
whereas pruned DP methods such as \texttt{EAPruned} reduce the amount of \emph{cell-level} computation once a candidate must be evaluated~\cite{herrmann_early_2021}.

Figure~\ref{fig_pruned_combo} demonstrates whether \textit{BGLB} remains useful in the presence of strong DP-level optimizations, where we compare the following pipelines:
(i) \texttt{EAPruned} alone, (ii) \textit{GLB}+\texttt{EAPruned}, and (iii) \textit{BGLB}+\texttt{EAPruned}.

We evaluate these pipelines on four representative UCR datasets that emphasize different computational challenges:  
\emph{ElectricDevices} and \emph{Crop} contain large numbers of time series (with training and test set sizes of $8926{+}7711$ and $7200{+}16800$, respectively),  
highlighting the advantage of candidate-level pruning in large search spaces.  
In contrast, \emph{Rock} and \emph{HouseTwenty} feature long sequences (lengths 2844 and 2000),  
which increase the per-candidate DP cost and thus amplify the benefit of cell-level pruning from \texttt{EAPruned}.
Figure~\ref{fig_pruned_combo} shows that \textit{BGLB} consistently delivers end-to-end speedups even when integrated with \texttt{EAPruned},
indicating that tighter candidate-level filtering remains valuable even when the remaining DP computations are aggressively pruned.

\begin{figure}[ht!]
    \centering
    \includegraphics[width=1.0\linewidth]{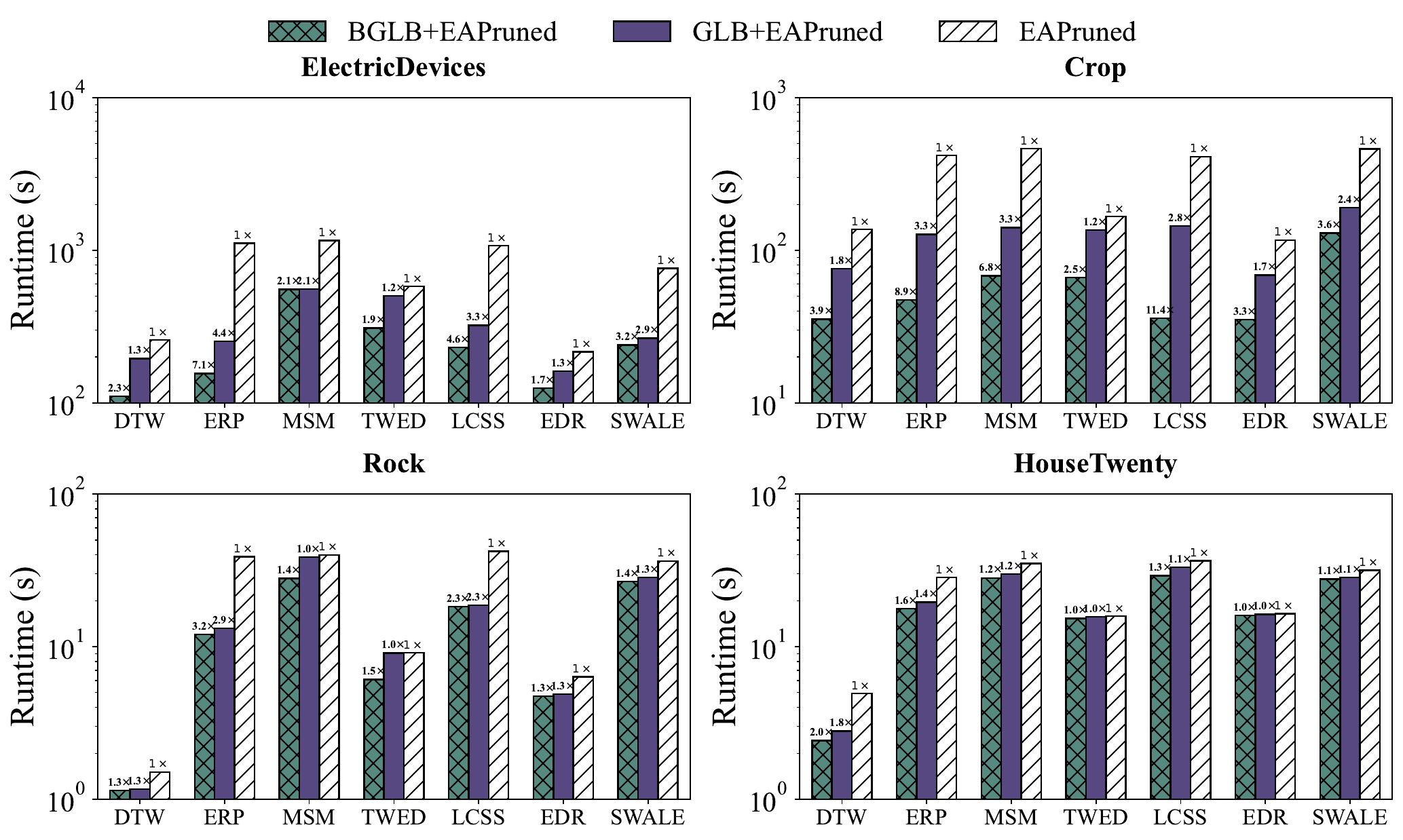}
    \caption{End-to-end runtime with pruned DP: \texttt{EAPruned} alone vs.\ \textit{GLB}+\texttt{EAPruned} vs.\ \textit{BGLB}+\texttt{EAPruned}. Lower bounds provide candidate-level filtering, while \texttt{EAPruned} reduces DP work for the remaining candidates.}
    \label{fig_pruned_combo}
\end{figure}

\subsection{Case Study: Accelerating Exact DBSCAN Clustering}
\label{sec:exp_dbscan}

We further evaluate the practical impact of tighter bounds on a database-style workload beyond 1-NN: exact DBSCAN clustering.
From a database perspective, this workload corresponds to an \emph{exact range query} operator under an elastic distance predicate.
\textit{BGLB} can be integrated as a \emph{plan-level filter} (a safe pre-check) inside query execution, reducing the number of expensive DP evaluations while preserving exact query semantics.
DBSCAN repeatedly issues $\epsilon$-neighborhood (range) queries, which can trigger a large number of elastic distance computations.
Lower bounds act as safe filters: if \(LB(X,Q)>\epsilon\), then \(ELD(X,Q)>\epsilon\) must hold, so the pair cannot be neighbors and can be discarded without affecting clustering correctness.
We use $\textit{minPts}=5$ by default and set $\epsilon$ via a simple measure-adaptive strategy.
Figure~\ref{fig_dbscan} reports end-to-end runtime on representative datasets.
Both \textit{GLB} and \textit{BGLB} significantly accelerate exact DBSCAN compared with the baseline that performs no lower-bound filtering,
and \textit{BGLB} consistently yields the fastest runtime, while preserving identical clustering results (exactness).

Beyond DBSCAN, the same filtering rule applies to any threshold-based operator: for any $\tau$, if $LB(X,Q)>\tau$ then $ELD(X,Q)>\tau$.
This makes the bound a drop-in safe pre-check for exact range queries, as well as for top-$k$ search when using the current $k$-th best-so-far distance as the dynamic threshold.



\begin{figure}[ht!]
    \centering
    \includegraphics[width=1.0\linewidth]{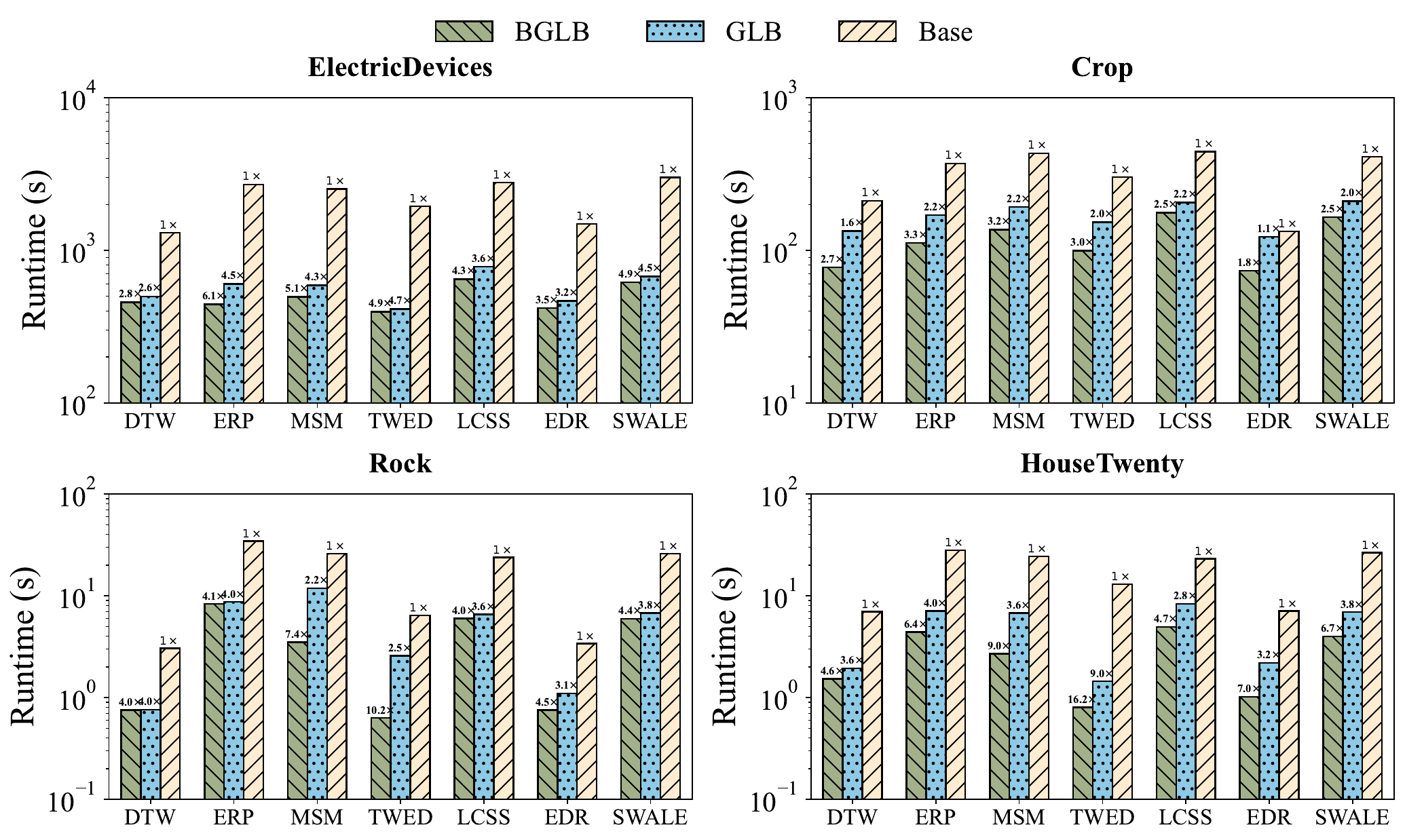}
    \caption{Case study on exact DBSCAN clustering: end-to-end runtime comparison between Base (no lower bound), \textit{GLB}-filtered, and \textit{BGLB}-filtered pipelines. Lower-bound filtering is safe and preserves exact clustering results.}
    \label{fig_dbscan}
\end{figure}

\subsection{Comparing with Specialized Lower Bounds on \textit{DTW}}
\label{sec:exp_dtw_specific}

\textit{DTW} is one of the most widely used elastic similarity measures, and a substantial body of work has developed highly specialized lower bounds for it.
Consequently, \textit{DTW} serves as a representative benchmark to assess how a general lower-bounding framework compares against DTW-specific designs.
Figure~\ref{fig_the_lower_dound_of_Elatic_distance_G} summarizes both tightness and computational cost, where the x-axis is the time to compute a lower bound and the y-axis is its TLB.
In the figure, the \emph{normalized} computation cost is reported as the runtime of each lower bound divided by the runtime of \emph{KimFL}, which is a constant-time baseline defined as $\mathrm{KimFL}(X,Q)=\max\{|x_1-q_1|,\;|x_n-q_m|\}$.
The results show that our method achieves a favorable trade-off between efficiency and pruning power among general-purpose approaches.
\begin{figure}[h]
	\centering
    \includegraphics[width=1\linewidth]{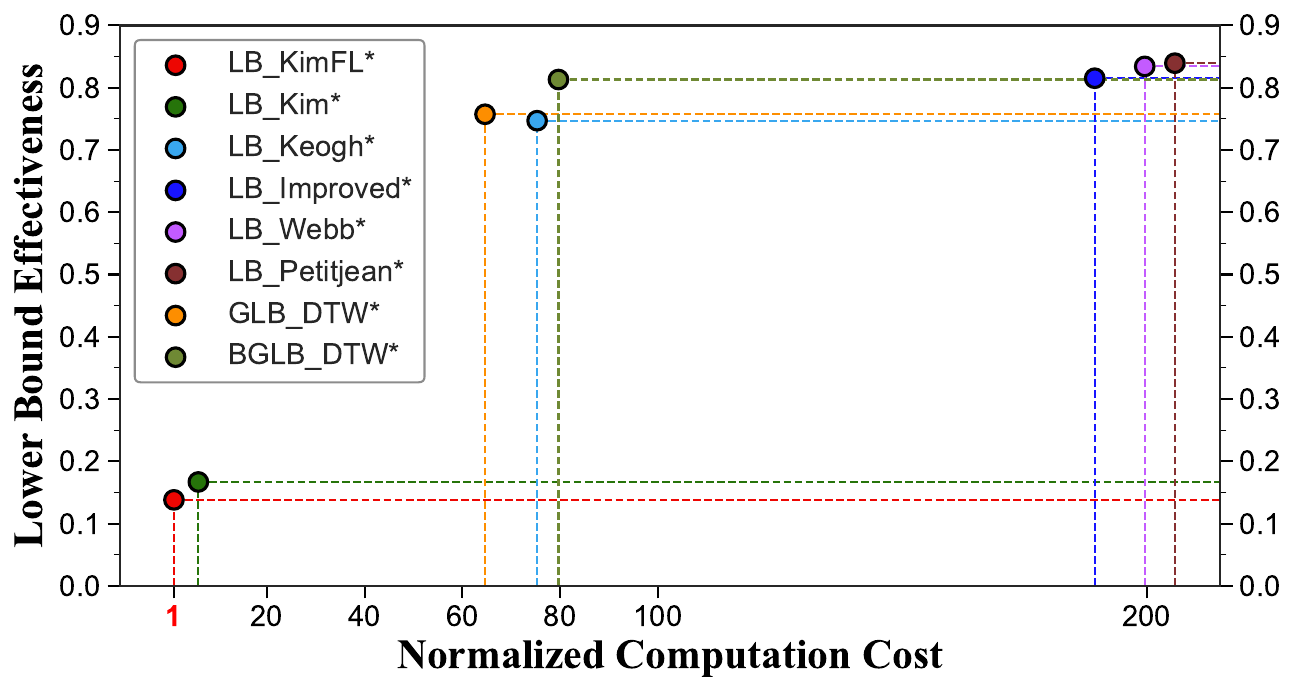}
    \caption{Effectiveness and efficiency of \textit{DTW} lower bounds.\protect\footnotemark}
    \label{fig_the_lower_dound_of_Elatic_distance_G}
\end{figure}

\footnotetext{For inherently symmetric bounds like $LB_{KimFL}$ and $LB_{Kim}$, we compute them only once without taking the maximum.}

\subsection{Impact of Parameters on Filtering Performance}
\label{sec:exp_window}

Finally, we investigate the impact of the warping window size \(\omega\) on the computational efficiency of 1-NN search.
Unless otherwise specified, we use $\omega=\mathrm{round}(0.05\,n)$; here we further vary $\omega$ to study sensitivity.
We set \(\omega\) to $round(0.05n)$, $round(0.1n)$, $round(0.2n)$, $round(0.5n)$, and $round(n)$, where $round(\cdot)$ denotes rounding to the nearest integer.
As shown in Figure~\ref{window}, the speedup of all methods decreases as \(\omega\) increases, since a larger warping window expands the set of feasible alignments and weakens pruning.
Nevertheless, our method consistently achieves greater speedup than existing approaches across all values of \(\omega\), indicating that the augmentation refinement remains beneficial even when the feasible alignment space becomes larger.

\begin{figure}[h]
	\centering
    \includegraphics[width=1\linewidth]{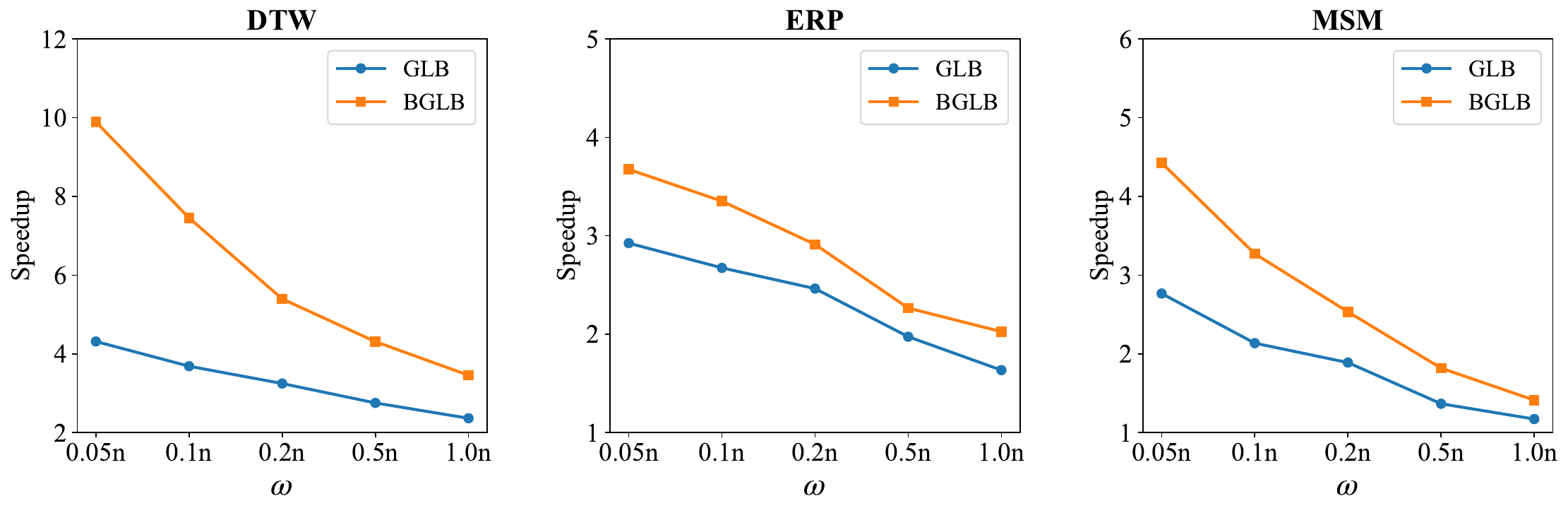}
    \caption{Impact of warping window size \(\omega\) on speedup in 1-NN similarity search.}
    \label{window}
\end{figure}

\section{Conclusion}

This paper proposes a new paradigm for deriving lower bounds of a broad class of elastic similarity measures, which formulate the lower-bounding problem via minimum-weight edge cover on an induced weighted bipartite graph.
Under this paradigm, we design \textit{BGLB}, which is a instantiation that incorporates an additional augmentation term.
Experiments on real-world datasets demonstrate that \textit{BGLB} achieves the tightest lower bounds for (ERP, MSM, TWED, LCSS, EDR and SWALE) and remains highly competitive for \textit{DTW}.
In 1-NN search experiments and DBSCAN, \textit{BGLB} achieves a significant speedup over the state-of-the-art approaches across seven elastic similarity measures.
In fact, \textit{BGLB} naturally fits into a wide range of time series applications, where the lower bound can be placed as an operator-level pre-check to reduce the number of expensive DP verifications while preserving exact semantics.

\begin{acks}
This work is supported by National Natural Science Foundation of China (NSFC) (62232005, 62202126) and China Postdoctoral Science Foundation  NO.2025M774320. We also thank the maintainers of the UCR Time Series Classification Archive for collecting, cleaning, and curating the 128 datasets used in our experiments.
\end{acks}

\bibliographystyle{ACM-Reference-Format}
\bibliography{sample}

@inproceedings{10.1145/3318464.3389760,
author = {Paparrizos, John and Liu, Chunwei and Elmore, Aaron J. and Franklin, Michael J.},
title = {Debunking Four Long-Standing Misconceptions of Time-Series Distance Measures},
year = {2020},
isbn = {9781450367356},
publisher = {Association for Computing Machinery},
address = {New York, NY, USA},
url = {https://doi.org/10.1145/3318464.3389760},
doi = {10.1145/3318464.3389760},
booktitle = {Proceedings of the 2020 ACM SIGMOD International Conference on Management of Data},
pages = {1887–1905},
numpages = {19},
location = {Portland, OR, USA},
series = {SIGMOD '20}
}

@inproceedings{10.1145/2339530.2339576,
author = {Rakthanmanon, Thanawin and Campana, Bilson and Mueen, Abdullah and Batista, Gustavo and Westover, Brandon and Zhu, Qiang and Zakaria, Jesin and Keogh, Eamonn},
title = {Searching and mining trillions of time series subsequences under dynamic time warping},
year = {2012},
isbn = {9781450314626},
publisher = {Association for Computing Machinery},
address = {New York, NY, USA},
url = {https://doi.org/10.1145/2339530.2339576},
doi = {10.1145/2339530.2339576},
booktitle = {Proceedings of the 18th ACM SIGKDD International Conference on Knowledge Discovery and Data Mining},
pages = {262–270},
numpages = {9},
location = {Beijing, China},
series = {KDD '12}
}

@article{10.14778/2735479.2735481,
author = {Ding, Rui and Wang, Qiang and Dang, Yingnong and Fu, Qiang and Zhang, Haidong and Zhang, Dongmei},
title = {YADING: fast clustering of large-scale time series data},
year = {2015},
issue_date = {January 2015},
publisher = {VLDB Endowment},
volume = {8},
number = {5},
issn = {2150-8097},
url = {https://doi.org/10.14778/2735479.2735481},
doi = {10.14778/2735479.2735481},
journal = {Proc. VLDB Endow.},
month = jan,
pages = {473–484},
numpages = {12}
}

@article{10.14778/1454159.1454226,
author = {Ding, Hui and Trajcevski, Goce and Scheuermann, Peter and Wang, Xiaoyue and Keogh, Eamonn},
title = {Querying and mining of time series data: experimental comparison of representations and distance measures},
year = {2008},
issue_date = {August 2008},
publisher = {VLDB Endowment},
volume = {1},
number = {2},
issn = {2150-8097},
url = {https://doi.org/10.14778/1454159.1454226},
doi = {10.14778/1454159.1454226},
journal = {Proc. VLDB Endow.},
month = aug,
pages = {1542–1552},
numpages = {11}
}

@inproceedings{10.1145/2020408.2020607,
author = {Kashyap, Shrikant and Karras, Panagiotis},
title = {Scalable kNN search on vertically stored time series},
year = {2011},
isbn = {9781450308137},
publisher = {Association for Computing Machinery},
address = {New York, NY, USA},
url = {https://doi.org/10.1145/2020408.2020607},
doi = {10.1145/2020408.2020607},
booktitle = {Proceedings of the 17th ACM SIGKDD International Conference on Knowledge Discovery and Data Mining},
pages = {1334–1342},
numpages = {9},
location = {San Diego, California, USA},
series = {KDD '11}
}

@article{10.1145/253262.253332,
author = {Korn, Flip and Jagadish, H. V. and Faloutsos, Christos},
title = {Efficiently supporting ad hoc queries in large datasets of time sequences},
year = {1997},
issue_date = {June 1997},
publisher = {Association for Computing Machinery},
address = {New York, NY, USA},
volume = {26},
number = {2},
issn = {0163-5808},
url = {https://doi.org/10.1145/253262.253332},
doi = {10.1145/253262.253332},
journal = {SIGMOD Rec.},
month = jun,
pages = {289–300},
numpages = {12}
}

@article{10.1145/2000824.2000827,
author = {Papapetrou, Panagiotis and Athitsos, Vassilis and Potamias, Michalis and Kollios, George and Gunopulos, Dimitrios},
title = {Embedding-based subsequence matching in time-series databases},
year = {2011},
issue_date = {August 2011},
publisher = {Association for Computing Machinery},
address = {New York, NY, USA},
volume = {36},
number = {3},
issn = {0362-5915},
url = {https://doi.org/10.1145/2000824.2000827},
doi = {10.1145/2000824.2000827},
journal = {ACM Trans. Database Syst.},
month = aug,
articleno = {17},
numpages = {39}
}

@INPROCEEDINGS{5693959,
  author={Camerra, Alessandro and Palpanas, Themis and Shieh, Jin and Keogh, Eamonn},
  booktitle={2010 IEEE International Conference on Data Mining}, 
  title={iSAX 2.0: Indexing and Mining One Billion Time Series}, 
  year={2010},
  volume={},
  number={},
  pages={58-67},
  doi={10.1109/ICDM.2010.124}}

@article{10.14778/2735471.2735476,
author = {Begum, Nurjahan and Keogh, Eamonn},
title = {Rare time series motif discovery from unbounded streams},
year = {2014},
issue_date = {October 2014},
publisher = {VLDB Endowment},
volume = {8},
number = {2},
issn = {2150-8097},
url = {https://doi.org/10.14778/2735471.2735476},
doi = {10.14778/2735471.2735476},
journal = {Proc. VLDB Endow.},
month = oct,
pages = {149–160},
numpages = {12}
}

@INPROCEEDINGS{7837992,
  author={Yeh, Chin-Chia Michael and Zhu, Yan and Ulanova, Liudmila and Begum, Nurjahan and Ding, Yifei and Dau, Hoang Anh and Silva, Diego Furtado and Mueen, Abdullah and Keogh, Eamonn},
  booktitle={2016 IEEE 16th International Conference on Data Mining (ICDM)}, 
  title={Matrix Profile I: All Pairs Similarity Joins for Time Series: A Unifying View That Includes Motifs, Discords and Shapelets}, 
  year={2016},
  volume={},
  number={},
  pages={1317-1322},
  doi={10.1109/ICDM.2016.0179}}

@article{10.14778/3551793.3551830,
author = {Paparrizos, John and Boniol, Paul and Palpanas, Themis and Tsay, Ruey S. and Elmore, Aaron and Franklin, Michael J.},
title = {Volume under the surface: a new accuracy evaluation measure for time-series anomaly detection},
year = {2022},
issue_date = {July 2022},
publisher = {VLDB Endowment},
volume = {15},
number = {11},
issn = {2150-8097},
url = {https://doi.org/10.14778/3551793.3551830},
doi = {10.14778/3551793.3551830},
journal = {Proc. VLDB Endow.},
month = jul,
pages = {2774–2787},
numpages = {14}
}

@inproceedings{breunig2000lof,
  title={LOF: identifying density-based local outliers},
  author={Breunig, Markus M and Kriegel, Hans-Peter and Ng, Raymond T and Sander, J{\"o}rg},
  booktitle={Proceedings of the 2000 ACM SIGMOD international conference on Management of data},
  pages={93--104},
  year={2000}
}

@article{10.1007/s10618-016-0483-9,
author = {Bagnall, Anthony and Lines, Jason and Bostrom, Aaron and Large, James and Keogh, Eamonn},
title = {The great time series classification bake off: a review and experimental evaluation of recent algorithmic advances},
year = {2017},
issue_date = {May       2017},
publisher = {Kluwer Academic Publishers},
address = {USA},
volume = {31},
number = {3},
issn = {1384-5810},
url = {https://doi.org/10.1007/s10618-016-0483-9},
doi = {10.1007/s10618-016-0483-9},
journal = {Data Min. Knowl. Discov.},
month = may,
pages = {606–660},
numpages = {55}
}

@article{10.14778/3342263.3342648,
author = {Paparrizos, John and Franklin, Michael J.},
title = {GRAIL: efficient time-series representation learning},
year = {2019},
issue_date = {July 2019},
publisher = {VLDB Endowment},
volume = {12},
number = {11},
issn = {2150-8097},
url = {https://doi.org/10.14778/3342263.3342648},
doi = {10.14778/3342263.3342648},
journal = {Proc. VLDB Endow.},
month = jul,
pages = {1762–1777},
numpages = {16}
}

@inproceedings{Sakoe1971ADP,
  title={A Dynamic Programming Approach to Continuous Speech Recognition},
  author={Hiroaki Sakoe and Seibi Chiba},
  year={1971},
  url={https://api.semanticscholar.org/CorpusID:107516844}
}

@ARTICLE{1163055,
  author={Sakoe, H. and Chiba, S.},
  journal={IEEE Transactions on Acoustics, Speech, and Signal Processing}, 
  title={Dynamic programming algorithm optimization for spoken word recognition}, 
  year={1978},
  volume={26},
  number={1},
  pages={43-49},
  doi={10.1109/TASSP.1978.1163055}}

@inproceedings{10.5555/1316689.1316758,
author = {Chen, Lei and Ng, Raymond},
title = {On the marriage of Lp-norms and edit distance},
year = {2004},
isbn = {0120884690},
publisher = {VLDB Endowment},
booktitle = {Proceedings of the Thirtieth International Conference on Very Large Data Bases - Volume 30},
pages = {792–803},
numpages = {12},
location = {Toronto, Canada},
series = {VLDB '04}
}

@ARTICLE{4479483,
  author={Marteau, Pierre-François},
  journal={IEEE Transactions on Pattern Analysis and Machine Intelligence}, 
  title={Time Warp Edit Distance with Stiffness Adjustment for Time Series Matching}, 
  year={2009},
  volume={31},
  number={2},
  pages={306-318},
  doi={10.1109/TPAMI.2008.76}}

@ARTICLE{6189346,
  author={Stefan, Alexandra and Athitsos, Vassilis and Das, Gautam},
  journal={IEEE Transactions on Knowledge and Data Engineering}, 
  title={The Move-Split-Merge Metric for Time Series}, 
  year={2013},
  volume={25},
  number={6},
  pages={1425-1438},
  doi={10.1109/TKDE.2012.88}}

@inproceedings{10.1145/1066157.1066213,
author = {Chen, Lei and \"{O}zsu, M. Tamer and Oria, Vincent},
title = {Robust and fast similarity search for moving object trajectories},
year = {2005},
isbn = {1595930604},
publisher = {Association for Computing Machinery},
address = {New York, NY, USA},
url = {https://doi.org/10.1145/1066157.1066213},
doi = {10.1145/1066157.1066213},
booktitle = {Proceedings of the 2005 ACM SIGMOD International Conference on Management of Data},
pages = {491–502},
numpages = {12},
location = {Baltimore, Maryland},
series = {SIGMOD '05}
}

@INPROCEEDINGS{994784,
  author={Vlachos, M. and Kollios, G. and Gunopulos, D.},
  booktitle={Proceedings 18th International Conference on Data Engineering}, 
  title={Discovering similar multidimensional trajectories}, 
  year={2002},
  volume={},
  number={},
  pages={673-684},
  doi={10.1109/ICDE.2002.994784}}

@inproceedings{10.1145/1247480.1247544,
author = {Morse, Michael D. and Patel, Jignesh M.},
title = {An efficient and accurate method for evaluating time series similarity},
year = {2007},
isbn = {9781595936868},
publisher = {Association for Computing Machinery},
address = {New York, NY, USA},
url = {https://doi.org/10.1145/1247480.1247544},
doi = {10.1145/1247480.1247544},
booktitle = {Proceedings of the 2007 ACM SIGMOD International Conference on Management of Data},
pages = {569–580},
numpages = {12},
location = {Beijing, China},
series = {SIGMOD '07}
}

@misc{UCRArchive,
title={The UCR Time Series Classification Archive},
author={ Chen, Yanping and Keogh, Eamonn and Hu, Bing and Begum, Nurjahan and Bagnall, Anthony and Mueen, Abdullah and Batista, Gustavo},
year={2015},
month={July},
note = {\url{www.cs.ucr.edu/~eamonn/time_series_data/}}
}

@INPROCEEDINGS{914875,
  author={Sang-Wook Kim and Sanghyun Park and Chu, W.W.},
  booktitle={Proceedings 17th International Conference on Data Engineering}, 
  title={An index-based approach for similarity search supporting time warping in large sequence databases}, 
  year={2001},
  volume={},
  number={},
  pages={607-614},
  doi={10.1109/ICDE.2001.914875}}

@article{10.5555/2993953.2994027,
author = {Keogh, Eamonn and Ratanamahatana, Chotirat Ann},
title = {Exact indexing of dynamic time warping},
year = {2005},
issue_date = {March 2005},
publisher = {Springer-Verlag},
address = {Berlin, Heidelberg},
volume = {7},
number = {3},
issn = {0219-1377},
journal = {Knowl. Inf. Syst.},
month = mar,
pages = {358–386},
numpages = {29}
}

@incollection{KEOGH2002406,
title = {Chapter 36 - Exact Indexing of Dynamic Time Warping},
editor = {Philip A. Bernstein and Yannis E. Ioannidis and Raghu Ramakrishnan and Dimitris Papadias},
booktitle = {VLDB '02: Proceedings of the 28th International Conference on Very Large Databases},
publisher = {Morgan Kaufmann},
address = {San Francisco},
pages = {406-417},
year = {2002},
isbn = {978-1-55860-869-6},
doi = {https://doi.org/10.1016/B978-155860869-6/50043-3},
url = {https://www.sciencedirect.com/science/article/pii/B9781558608696500433},
author = {Eamonn Keogh}
}

@inbook{doi:10.1137/1.9781611975321.27,
author = {Yilin Shen and Yanping Chen and Eamonn Keogh and Hongxia Jin},
title = {Accelerating Time Series Searching with Large Uniform Scaling},
booktitle = {Proceedings of the 2018 SIAM International Conference on Data Mining (SDM)},
chapter = {},
pages = {234-242},
doi = {10.1137/1.9781611975321.27},
URL = {https://epubs.siam.org/doi/abs/10.1137/1.9781611975321.27},
eprint = {https://epubs.siam.org/doi/pdf/10.1137/1.9781611975321.27}
}

@article{LEMIRE20092169,
title = {Faster retrieval with a two-pass dynamic-time-warping lower bound},
journal = {Pattern Recognition},
volume = {42},
number = {9},
pages = {2169-2180},
year = {2009},
issn = {0031-3203},
doi = {https://doi.org/10.1016/j.patcog.2008.11.030},
url = {https://www.sciencedirect.com/science/article/pii/S0031320308004925},
author = {Daniel Lemire}
}

@article{WEBB2021107895,
title = {Tight lower bounds for dynamic time warping},
journal = {Pattern Recognition},
volume = {115},
pages = {107895},
year = {2021},
issn = {0031-3203},
doi = {https://doi.org/10.1016/j.patcog.2021.107895},
url = {https://www.sciencedirect.com/science/article/pii/S0031320321000820},
author = {Geoffrey I. Webb and François Petitjean}
}

@inproceedings{10.1145/956750.956777,
author = {Vlachos, Michail and Hadjieleftheriou, Marios and Gunopulos, Dimitrios and Keogh, Eamonn},
title = {Indexing multi-dimensional time-series with support for multiple distance measures},
year = {2003},
isbn = {1581137370},
publisher = {Association for Computing Machinery},
address = {New York, NY, USA},
url = {https://doi.org/10.1145/956750.956777},
doi = {10.1145/956750.956777},
booktitle = {Proceedings of the Ninth ACM SIGKDD International Conference on Knowledge Discovery and Data Mining},
pages = {216–225},
numpages = {10},
location = {Washington, D.C.},
series = {KDD '03}
}

@article{Vlachos2006Indexing,
  author    = {Vlachos, Michail and Hadjieleftheriou, Marios and Gunopulos, Dimitrios and Keogh, Eamonn},
  title     = {Indexing Multidimensional Time-Series},
  journal   = {The VLDB Journal},
  volume    = {15},
  pages     = {1--20},
  year      = {2006},
  doi       = {10.1007/s00778-004-0144-2},
  url       = {https://doi.org/10.1007/s00778-004-0144-2}, 
}

@article{Tan2020FastEE,
  author    = {Tan, Chang Wei and Petitjean, François and Webb, Geoffrey I.},
  title     = {FastEE: Fast Ensembles of Elastic Distances for time series classification},
  journal   = {Data Mining and Knowledge Discovery},
  volume    = {34},
  pages     = {231--272},
  year      = {2020},
  doi       = {10.1007/s10618-019-00663-x},
  url       = {https://doi.org/10.1007/s10618-019-00663-x}
}

@inproceedings{10.5555/3000850.3000887,
author = {Berndt, Donald J. and Clifford, James},
title = {Using dynamic time warping to find patterns in time series},
year = {1994},
publisher = {AAAI Press},
booktitle = {Proceedings of the 3rd International Conference on Knowledge Discovery and Data Mining},
pages = {359–370},
numpages = {12},
location = {Seattle, WA},
series = {AAAIWS'94}
}

@article{10.1145/2949741.2949758,
author = {Paparrizos, John and Gravano, Luis},
title = {k-Shape: Efficient and Accurate Clustering of Time Series},
year = {2016},
issue_date = {March 2016},
publisher = {Association for Computing Machinery},
address = {New York, NY, USA},
volume = {45},
number = {1},
issn = {0163-5808},
url = {https://doi.org/10.1145/2949741.2949758},
doi = {10.1145/2949741.2949758},
journal = {SIGMOD Rec.},
month = jun,
pages = {69–76},
numpages = {8}
}

@inbook{10.5555/108235.108244,
author = {Sakoe, Hiroaki and Chiba, Seibi},
title = {Dynamic programming algorithm optimization for spoken word recognition},
year = {1990},
isbn = {1558601244},
publisher = {Morgan Kaufmann Publishers Inc.},
address = {San Francisco, CA, USA},
booktitle = {Readings in Speech Recognition},
pages = {159–165},
numpages = {7}
}

@article{10.14778/3594512.3594530,
author = {Paparrizos, John and Wu, Kaize and Elmore, Aaron and Faloutsos, Christos and Franklin, Michael J.},
title = {Accelerating Similarity Search for Elastic Measures: A Study and New Generalization of Lower Bounding Distances},
year = {2023},
issue_date = {April 2023},
publisher = {VLDB Endowment},
volume = {16},
number = {8},
issn = {2150-8097},
url = {https://doi.org/10.14778/3594512.3594530},
doi = {10.14778/3594512.3594530},
journal = {Proc. VLDB Endow.},
month = apr,
pages = {2019–2032},
numpages = {14}
}

@INPROCEEDINGS{7354388,
  author={Abboud, Amir and Backurs, Arturs and Williams, Virginia Vassilevska},
  booktitle={2015 IEEE 56th Annual Symposium on Foundations of Computer Science}, 
  title={Tight Hardness Results for LCS and Other Sequence Similarity Measures}, 
  year={2015},
  volume={},
  number={},
  pages={59-78},
  doi={10.1109/FOCS.2015.14}}

@article{DBLP:journals/pvldb/ChaoZQW25,
  author       = {Zemin Chao and
                  Qiaoyi Zheng and
                  Zhixin Qi and
                  Hongzhi Wang},
  title        = {{FSMDTW:} {A} Fast Index-free Subsequence Matching Algorithm for Dynamic
                  Time Warping},
  journal      = {Proc. {VLDB} Endow.},
  volume       = {18},
  number       = {10},
  pages        = {3628--3640},
  year         = {2025},
  url          = {https://www.vldb.org/pvldb/vol18/p3628-wang.pdf},
  bibsource    = {dblp computer science bibliography, https://dblp.org}
}

@article{DBLP:journals/tkde/ChaoGMLW25,
  author       = {Zemin Chao and
                  Hong Gao and
                  Dongjing Miao and
                  Jianzhong Li and
                  Hongzhi Wang},
  title        = {An Amortized {O(1)} Lower Bound for Dynamic Time Warping in Motif
                  Discovery},
  journal      = {{IEEE} Trans. Knowl. Data Eng.},
  volume       = {37},
  number       = {5},
  pages        = {2239--2252},
  year         = {2025},
  url          = {https://doi.org/10.1109/TKDE.2025.3544751},
  doi          = {10.1109/TKDE.2025.3544751},
  bibsource    = {dblp computer science bibliography, https://dblp.org}
}

@article{keogh2001dimensionality,
  title={Dimensionality reduction for fast similarity search in large time series databases},
  author={Keogh, Eamonn and Chakrabarti, Kaushik and Pazzani, Michael and Mehrotra, Sharad},
  journal={Knowledge and information Systems},
  volume={3},
  pages={263--286},
  year={2001},
  publisher={Springer}
}

@article{herrmann_early_2021,
	title = {Early abandoning and pruning for elastic distances including dynamic time warping},
	volume = {35},
	issn = {1573-756X},
	url = {https://doi.org/10.1007/s10618-021-00782-4},
	doi = {10.1007/s10618-021-00782-4},
	number = {6},
	journal = {Data Mining and Knowledge Discovery},
	author = {Herrmann, Matthieu and Webb, Geoffrey I.},
	month = nov,
	year = {2021},
	pages = {2577--2601},
}

\end{document}